\documentclass[11pt,a4paper]{article}

\usepackage{array}

\usepackage[T1]{fontenc}
\usepackage{a4wide}
\usepackage{mathrsfs}
\usepackage{amsmath,amsthm,amssymb}
\usepackage{latexsym}
\usepackage{enumerate}
\usepackage{mathrsfs}
\usepackage{stmaryrd}
\usepackage{amsopn}
\usepackage{amsmath}
\usepackage{amssymb,bm}
\usepackage{amsfonts}
\usepackage{amsbsy}
\usepackage{amscd,indentfirst,epsfig}
\usepackage{color}
\usepackage{hyperref}

\usepackage{enumerate}
\usepackage{mathrsfs}
\usepackage{amsopn}
\usepackage{amsmath}
\usepackage{amssymb}
\usepackage{amsfonts}
\usepackage{amsbsy}
\usepackage{amscd,indentfirst,epsfig}
\usepackage{amsfonts,amsmath,latexsym,amssymb,verbatim,amsbsy}
\usepackage{amsthm}
\usepackage{colordvi}
\numberwithin{equation}{section}

\usepackage[affil-it]{authblk}

\def \dis {\displaystyle}

\def \P {{\Psi}}
\def\R{\mathbb R}
\def\S{{\mathbb S}^2}

\def\N{\mathbb N}

\def\C{\mathcal C}

\def\D{\mathscr{D}}

\def\M{\mathcal M}

\def\E{\mathcal E}
\def \mS {\mathcal{S}}
\def \B{\mathsf{B}}

\def\d{\mathrm{d}}

\def\Q{\mathcal{Q}}

\def \V {\bm{V}_{N}}
\def \Vk {\bm{V}_{\ell}}

\def \Vku {\bm{V}_{\ell+1}}
\def \le {\leqslant}
\def \ge {\geqslant}

\def \leq {\leqslant}
\def \geq {\geqslant}

\def \Gk {\bm{\Gamma}_{k;k+1}}

\newtheorem{theo}{Theorem}[section]
\newtheorem{cor}[theo]{Corollary}
\newtheorem{lem}[theo]{Lemma}
\newtheorem{prp}[theo]{Proposition}
\newtheorem{rem}[theo]{Remark}

\newtheorem{ass}[theo]{Assumption}
\newtheorem{conj}[theo]{Conjecture}
\theoremstyle{definition}
\newtheorem{defi}[theo]{Definition}
\theoremstyle{example}

\begin{document}

\title{\textbf{A Kac model for kinetic annihilation}}

\author[1]{Bertrand \textsc{Lods} \thanks{\texttt{bertrand.lods@unito.it }}}
\author[2]{Alessia \textsc{Nota} \thanks{\texttt{nota@iam.uni-bonn.de }}}
\author[3]{Federica \textsc{Pezzotti} \thanks{\texttt{pezzotti@mat.uniroma1.it }}}

\affil[1] {\em Dipartimento ESOMAS, Universit\`a degli Studi di Torino \& Collegio Carlo Alberto, \em Corso Unione Sovietica, 218/bis, 10134 Torino, Italy}
\affil[2]{\em Institute for Applied Mathematics, University of Bonn, \em Endenicher Allee 60, D-53115 Bonn, Germany}  
\affil[3]{\em Dipartimento di Matematica "G. Castelnuovo", Sapienza, Universit\`{a} di Roma, Italy}

\date{\today}

\maketitle

\begin{abstract}
In this paper we consider the stochastic dynamics of a finite system of particles in a finite volume (Kac-like particle system) which annihilate with probability $\alpha \in (0,1)$ or collide elastically with probability $1-\alpha$. We first establish the well-posedness of the particle system which exhibits no conserved quantities. We rigorously prove that, in some thermodynamic limit, a suitable 
hierarchy of kinetic equations  is recovered for which tensorized solution to the homogenous Boltzmann with annihilation is a solution. For bounded collision kernels, this shows in particular that propagation of chaos holds true. Furthermore, we make conjectures about the limit behaviour of the particle system when hard-sphere interactions are taken into account. \end{abstract}

\tableofcontents

 \section{Introduction}\label{sec:intro}

\subsection{The kinetic annihilation equation}

In a kinetic framework, the behavior of a system of particles 
which annihilate with probability $\alpha \in (0,1)$ or collide elastically with probability $1-\alpha$ can be described (in a spatially homogeneous situation) by the so-called velocity distribution $f(t,v)$ which represents the probability density of particles with velocity $v \in \R^d$ $(d \geq 2)$ at time $t \geq 0.$  The time-evolution of the
one-particle distribution function \(f(t,v)\), \(v\in\R^d\),
\(t>0\) satisfies the following
\begin{equation}\label{BE}
\partial_t f(t,v)=(1-\alpha)\Q(f,f)(t,v) -\alpha \Q_-(f,f)(t,v)
\end{equation}
where $\Q=\Q_{+}-\Q_{-}$ is the quadratic Boltzmann collision operator. The gain part $\Q_{+}$ and loss part $\Q_{-}$ are defined by the bilinear symmetrized forms
 \begin{multline*}\label{bilin}
 \Q_{+}(g,f)(v) = \frac{1}{2}\,\int _{\R^3 \times \S } \B(v-v_*,\omega)
         \left(g'_* f' + g' f_* '\right) \, \d v_* \, \d\omega, \\
  \Q_{-}(g,f)(v) = \frac{1}{2}\,\int _{\R^3 \times \S } \B(v-v_*,\omega)
         \left(g_* f  + g  f_*  \right) \, \d v_* \d\omega,
 \end{multline*}
where we have used the shorthands $f=f(v)$, $f'=f(v')$, $g_*=g(v_*)$ and
$g'_*=g(v'_*)$ with post-collisional velocities $v'$ and $v'_*$  parametrized by
 \begin{equation}\label{eq:rel:vit}
 v'=v- \left[(v-v_{*}) \cdot \omega\right] \omega \qquad \text{ and } \quad v'_*=v_* + \left[(v-v_{*}) \cdot \omega\right] \omega, \qquad \omega   \in  {\S}
.\end{equation}

The equation above has been introduced recently in \cite{Ben-Naim,coppex04, coppex05, Kaprivsky, Piasecki,Trizac} as a peculiar kinetic model aiming to test the relevance of non-equilibrium statistical mechanics for systems of \textit{reacting particles}. Such systems of reacting particles are particularly challenging in particular because the lack of collisional invariants make the derivation of suitable hydrodynamics non trivial, even at a formal level. Notice indeed that the kinetic equation \eqref{BE} is \emph{highly disspative} since any reasonable solution $f(t,v)$  to \eqref{BE} is such that its mass and kinetic energy
$$n(t)=\int_{\R^{3}}f(t,v)\d v, \qquad E(t)=\int_{\R^{3}}f(t,v)|v|^{2}\d v$$
are decreasing in time, i.e. 
\begin{equation}\label{eq:denE}
\dfrac{\d}{\d t}n(t)=-\alpha \int_{\R^d}\Q_-(f,f)(t,v)\d v \leq 0, \qquad \text{ while } \quad \dfrac{\d}{\d t}E(t)=-\alpha \int_{\R^d} |v|^2\Q_-(f,f)(t,v)\d v \leq 0.\end{equation}
Because of this,  the only possible long time behaviour for $f(t,v)$ is  
$$\lim_{t\to \infty}f(t,v)=0.$$
The long time behavior of solutions to \eqref{BE} has been studied, at a mathematical level, in the recent papers \cite{BaLo,BaLo2,ABL} providing the existence, uniqueness and stability of suitable self-similar profile associated to \eqref{BE} which captures its asymptotic behaviour for long time in a more accurate way.  

With the aforementioned contributions, the qualitative behaviour of solutions to \eqref{BE} is by now quite well understood but a rigorous derivation of the equation from a physically grounded model of interacting particles is still missing. Filling this blank is the purpose of the present paper. 

\subsection{Justification of the kinetic model  \eqref{BE}: Kac-like system}\label{sec:kac}

As we discussed in the previous section, the goal of this paper is to provide a physical ground to the spatially homogeneous Boltzmann equation with annihilation \eqref{BE}. In order to do this, we provide an approximating Kac-like particle system from which we aim to recover \eqref{BE} in a suitable scaling limit where the number of particles, as well as the volume, goes to infinity but the number density is finite. 

In particular, one of the main question in the justification of the model is to prove the validity of the so-called ``Boltzmann's molecular chaos assumption'' (\emph{Stosszahlansatz}). More specifically, the unkown $f(t,v)$ of Eq. \eqref{BE} represents the  distribution of a typical particle  with velocity $v$ at time $t \geq0 $ and, as well-known, the right-hand side of \eqref{BE} is a bilinear operator which models the change of velocity of a particle  -- due to collisions with another particle or annihilation. In particular, the loss part $\Q_{-}$ of the collision operator  takes into account the disappearance of particles with velocity $v$ due to the encounter with particles with velocities $v_{*}$. Typically,  the involved distribution appearing there, should be the joint distribution $f_{2}(v,v_{*})$ of finding 2 particles with velocities $(v,v_{*})$. In \eqref{BE}, this joint distribution is replaced by the tensorized distribution $f(v) \otimes f(v_{*})$ which is related to some no-correlation assumption. 

Proving the propagation of chaos is roughly speaking, justifying this no-correlation assumption and it is a very challenging problem for collisional kinetic equation whose story can be traced back to Boltzmann himself. The literature on the subject is extremely wide: we mention the seminal paper by O. E. Lanford \cite{lanford} concerning the derivation of the spatially inhomogeneous Boltzmann equation from a particles system in the \emph{Boltzmann-Grad limit}, see also the more recent contributions \cite{GTSR,chiara}. We also mention the recent analysis of the correlation error for Boltzmann equation in \cite{simonella} which, as \cite{lanford}, works in the \emph{\textbf{grand canonical ensemble formalism}} which is the one adopted here. 

For the spatially homogeneous case the justification of kinetic equations  from a particles system is usually performed through a so-called \emph{\textbf{mean-field limit}} following the program suggested by M. Kac who, in a seminal paper \cite{K}, proposed  a first attempt of clarifying some aspects of the transition from an
$N$-particle system to a one-particle kinetic description. In this paper (see also the modification of the model in \cite{kean}), he introduced a fundamental \emph{stochastic particle model} which  consists in a system of $N$
particles with associated velocities $\bm{V}_N=(v_1, \dots,v_N)\in\R^{3N}$, whose dynamics is the following stochastic process: 
at a random time chosen accordingly to a Poisson process, pick a pair of particles, say $i$ and $j$, and perform
the transition 
$$v_i,v_j\to v^{\prime}_i,v^{\prime}_j,$$
where $v^{\prime}_{i},v^{\prime}_{j}$ are given in \eqref{eq:vivjprime}.  The forward Kolmogorov equation associated to this Kac's stochastic model is then called the master equation whose unknown $F_{N}(t,\bm{V}_{N})$ represents the density of particles having the velocities $\bm{V}_{N} \in \R^{3N}$ at time $t \geq 0$ and the proof of propagation of chaos for the spatially homogeneous Boltzmann equation consists in showing that marginals of $F_{N}(t,V_{N})$ converge (as $N \to \infty$) to tensorization of the solution to the Boltzmann equation. This question has been addressed in a series of papers, the first ones dealing with bounded collision kernels and  culminating in the contribution \cite{MM} where the propagation of chaos was obtained in a quantitative way for general initial data and both hard-sphereand maxwellian molecules. We refer to \cite{MM} for an account of the literature on the subject. We just aim to mention here that suitable modifications of original Kac's stochastic model have been proposed in the literature to handle different kinds of kinetic-like equations, including the Landau equation \cite{fournier, miot} or models with quantum interactions \cite{fede}. The literature for mean field limit for Vlasov-like equations (including $2D$ Euler vortex model) is even more abundant, we mention here only \cite{spohn, jabin, mischler} among important contributions to the field.

To take into account the annihilation of particles, we introduce here a modification of the original Kac model which, in particular, includes the possibility of mass dissipation.  We can already mention here that the scaling limit we shall perform is not stictly speaking of \textit{mean-field} type but has rather to be seen as a \textit{\textbf{thermodynamic limit}}.  
More precisely, we consider an $N$-particle system in a region $\mathcal{D}$ with finite volume $\Lambda=|\mathcal{D}|$ whose state space is $\R^{3N}$. The evolution is the following. 
Suppose that we have two clocks that at an exponential time pick a pair of particles. At a first collection of times $\{t_k\}_{k\geq 1}$  which are 
separated exponentially at rate $r_1=(1-\alpha)r$ 
with independent increments $t_k-t_{k-1}$ two particles (say $i$ and $j$) are chosen, uniformly and at random, to collide. This gives the jump process 
 $\bm{V}_N=\{v_1,\dots, v_i,\dots,v_j,\dots,v_N\}\to \bm{V}^{i,j}_N=\{v_1,\dots, v^{\prime}_i,\dots,v^{\prime}_j,\dots,v_N\}$ 
where 
\begin{equation}\label{eq:vivjprime}
v_i'=v_i- [ (v_i-v_j)\cdot \omega] \omega \qquad \text{ and } \quad v'_j=v_j +[ (v_i-v_j)\cdot \omega]  \omega\end{equation}
are the outgoing velocities arising from an elastic collision with scattering vector $\omega$. 
While, at a second collection of times $\{\tilde{t}_k\}_{k\geq 1}$ which are 
separated exponentially at rate $r_2=\alpha r$ two particles (say $i$ and $j$) are chosen uniformly and at random annihilate, disappearing from the system. 

The probability of such a transition is assumed to be a function of the modulus of the 
relative velocities of the two particles involved in the collision and of the angle between their relative velocity and the unit scattering vector $\omega\in \S$ (see Section \ref{sec:colker}). 

The \emph{master equation} for this model is the Kolmogorov equation associated to the Markov
process we are considering. We now describe it with more details. Since we want to describe a finite system of particles, contained in a finite volume $\mathcal{D}\subset \R^3$, whose number $N$ is  \textit{not fixed through time}, it appears  convenient to use the \emph{\textbf{grand canonical ensemble}} formalism for which the sample space is defined as
\begin{equation}
\Omega=\bigcup_{N=0}^{\infty}\left (N, \mathscr{P}_N(\mathbb{R}^3)\right)\subset  \mathbb{N}\times  \mathscr P({\mathbb R}^3),  \label{S2E1}\end{equation}
where $\mathscr P_N(\mathbb{R}^3)$ is the family of subsets of $\mathbb{R}^3$ with cardinality $N$. The sample set $\Omega$ is therefore the set of all the pairs $(N, \omega _N)$ where $N\in \mathbb{N}$ and $\omega _N$ is any finite subset of $N$ three dimensional vectors  $v_1, v_2, \cdots, v_{N}$.  For every time $t\ge 0$ the probability distribution of each state $(N, (v_1, \dots,  v_{N}))\in \Omega$ is 
$$\dfrac{1}{N!}\P_N(t, v_1, \dots, v_{N})$$ where $\left\{\P_N(t)\right\}_{N\in  \mathbb{N}}$ is a sequence of non negative functions $\P_N(t)=\P_N(t, \cdot)$, each of them  defined  on $\mathscr P_N({\mathbb R}^3)$ and normalized according to
\begin{eqnarray}
\label{S2E2}
\sum_{N=0}^\infty \frac{1}{N!}\int_{\R^{3}}\d  v_1\cdots \int_{\R^{3}}\d v_{N}  \P_N(t, v_1, \cdots, v_{N})=1.
\end{eqnarray}
The functions $\P_N(t, v_1, \dots, v_N)$ are assumed to be symmetric with respect to any permutation of the indices  $1,\dots, N$ and no  restrictions are imposed on  the range of velocity values.\medskip

For any $N \in \N$ and $\V=(v_1,\ldots,v_{N}) \in \R^{3N}$, $t \geq 0$,  $\P_N(t,\V)=\P_N(t, v_1, \dots, v_{N})$ is 
the velocity distribution  function of the $N$-particle configuration $(v_1, \dots, v_{N})$. The evolution equation for the velocity distribution function $\P_N(t,\V)$, 
i.e.~the forward Kolmogorov equation for the stochastic process described above, is the following master equation:
\begin{align}\label{PNVN}
\partial_t \P_N(t,\V) =& \, (1-\alpha)\sum_{1 \leq i < j \leq N}\int_{\S} \B_{\Lambda}(v_i-v_j,\omega) \left[\P_N\left(t,\V^{i,j}\right)-\P_N\left(t,\V\right)\right]\d\omega \nonumber \\&
+\frac{\alpha}{2} \int_{\R^3} \d v_{N+1}\int_{\R^3}\d v_{N+2}\int_{\S} \B_{\Lambda}(v_{N+1}-v_{N+2},\omega)\P_{N+2}\left(t,\V, v_{N+1},v_{N+2}\right)\d\omega \nonumber \\&
 -\alpha \sum_{1 \leq i < j \leq N}\int_{\S} \B_{\Lambda}(v_i-v_j,\omega)\P_N\left(t,\V\right)\d\omega,
\end{align}
supplemented with initial datum $\P_N(0,\V)\in L^1(\R^{3N})$. Here, for any scattering vector $\omega \in \S$, we denoted by
$$\V^{i,j}=(v_1,\ldots,v_{i-1},v_i',v_{i+1},\ldots,v_{j-1},v_j',v_{j+1},\ldots,v_{N})$$
where
$$v_i'=v_i- [ (v_i-v_j)\cdot \omega] \omega \qquad \text{ and } \quad v'_j=v_j +[ (v_i-v_j)\cdot \omega]  \omega$$
are the post-collisional velocities.

The collision kernel $\B_{\Lambda}$ is a suitable modification of the kernel $\B$ in \eqref{BE} which takes into account the fact that, in \eqref{PNVN}, particles are enclosed in a   bounded region $\mathcal{D} \subset \R^{3}$ with finite volume $\Lambda=|\mathcal{D}|$ whereas, in \eqref{BE}, they are distributed in the whole space $\R^{3}$ (see Assumption \ref{hyp2}). More precisely, even if both \eqref{BE} and \eqref{PNVN} are spatially homogeneous, spatial effects are implicitly taken into account: the spatially homogeneous assumption in \eqref{PNVN}  means that the distribution function $\P_{N}(t,x_{1},\ldots,x_{N},v_{1},\ldots,v_{N})$ of $N$ particles at time $t \geq0$ having position $(x_{1},\ldots,x_{N}) \in \mathcal{D}^{N}$ with velocities $v_{1},\ldots,v_{N} \in \R^{3N}$ is the same for all $x_{i} \in \mathcal{D}$ which allows to drop the dependency with respect to the space variables $x_{1},\ldots,x_{N}$. Similar considerations hold for \eqref{BE}.

\begin{rem}\label{rem:alpha=0}Notice that the first term on the right hand side of \eqref{PNVN} is the usual term of the standard Kac master equation. This means that, neglecting the annihilation, i.e. for $\alpha=0$, we would obtain only this term. The third one on the right hand side of \eqref{PNVN} is the loss term due to particle annihilation, while the second one takes into account the gain term due to particle annihilation and we have the factor $\frac{1}{2}$ to avoid double counting.
\end{rem} 

\subsection{About the collision kernel and the initial distribution} \label{sec:colker}

Our general assumption on the collision kernel $\B$ appearing in \eqref{BE} is the following
\begin{ass}\label{hyp1} 
The collision kernel $\B(\cdot,\cdot)$ is a measure and nonnegative mapping $\B\: :\:\R^{3} \times \S \to \R^{+}$ for which there exist $\gamma \in [0,1]$ and $C_{\B} >0$ such that
\begin{equation}\label{eq:growth}
\Sigma_{\B}(v-v_{*}) \leq C_{\B}|v-v_{*}|^{\gamma}, \qquad \forall v,v_{*} \in \R^{3}\end{equation}
where we denote by $\Sigma_{\B}$ the {collision frequency}: 
$$\Sigma_{\B}(z)=\int_{\S}\B(z,\omega)\d\omega, \qquad z \in \R^{3}.$$
\end{ass}

\begin{rem} The case
$$\B(v-v_*,\omega)=2\,|\S|^{-1}\,|(v-v_{*}) \cdot \omega|=\tfrac{1}{2\pi}|(v-v_{*}) \cdot \omega |$$ corresponding to hard-sphere interactions is the model usually considered in the physics literature \cite{Maynar1,Maynar2,Trizac}. Here and in \eqref{eq:rel:vit}, the dot symbol $\cdot$ denotes the usual inner product between three dimensional vectors. In this case
$$\Sigma_{B}(v-v_{*})=|v-v_{*}|, \qquad v,v_{*} \in \R^{3}.$$

The case $\Sigma_{\B}$ constant, say $\Sigma_{B}(v-v_{*})=1$,  corresponds to the so-called Maxwellian interactions for which $\gamma=0$ and 
$$\B(v-v_{*},\omega)=\frac{1}{2\pi}\frac{\left|(v-v_{*})\cdot \omega\right|}{|v-v_{*}|}, \qquad v \neq v_{*} \in \R^{3}, \quad \omega \in \R^{3}.$$
Notice that, in this case, the kinetic annihilation equation \eqref{BE} is \emph{equivalent} to the classical Boltzmann equation (for which $\alpha=0$), see \cite{BaLo} and references therein for details. \end{rem}
\begin{rem} In the sequel, for most of the paper, 
we will consider a general collision kernel $\B$ satisfying Assumption \ref{hyp1}. We will have to restrict to \emph{bounded} collision frequency $\Sigma_{\B}$ only in the last part of the paper to recover the uniqueness of the solution to the \emph{annihilated Boltzmann hierarchy} and, as such, the propagation of chaos for \eqref{BE}. This is a severe restriction but we notice that assuming $\Sigma_{{\B}}$ to be bounded, i.e. $\gamma=0$, we cover more general situation than the Maxwellian interactions case.
\end{rem}

We observe that, under the growth condition \eqref{eq:growth}, it is possible to show that the model \eqref{PNVN} is well posed (cf. Section \ref{sec:many}, Theorem \ref{th:wellposPS}) under suitable assumptions on the collision kernel $\B$ and on the initial data. More precisely, we will assume the following volume dependence of the collision kernel:
\begin{ass}\label{hyp2} 
Let be $\Lambda=|\mathcal{D}|> 0$ with $\mathcal{D} \subset \R^{3}$. We set 
$$\B_{\Lambda}(z,\omega)=\frac{1}{\Lambda}\B(z,\omega)\qquad \forall z \in \R^{3}, \omega \in \S,$$
where $\B$ is the collision kernel satisfying Assumption \ref{hyp1}.
\end{ass}

The volume dependence of $\B_{\Lambda}$ required in Assumption \ref{hyp2}  takes into account the fact that the collision rate decreases as the proportion of the volume
occupied  by  the  particles,  with  respect  to  the  total  volume $\Lambda$, increases. We adopted  the simplest scenario for which the  dependence is inversely  proportional  to $\Lambda$, namely $\frac 1 \Lambda$.

We now make precise our assumption on the initial datum for \eqref{PNVN}: to justify the propagation of chaos, we will start with well-prepared initial data which are already tensorized:
\begin{ass}\label{hyp:indata}  
We assume that the initial datum $\{\P_N(0)\}_{N \geq 1}$ is given by
\begin{equation}\label{inDATA}
\P_N(0,\V)=
\begin{cases}
N_0 ! f_0(v_1)\ldots f_0(v_{N_0})=:N_0 ! f_0^{\otimes_{N_0}}(V_{N_0}) \qquad &\text{ if  } N=N_0\\
0 \qquad &\text{ if } N \neq N_0
\end{cases}
\end{equation}
for some $N_0 \geq 1$ and some non negative probability distribution $f_0$ satisfying
\begin{equation}\label{normalization}
\int_{\R^3}f_0(v)\d v=1, \qquad \int_{\R^{3}}|v|^{2}f_{0}(v)\d v=E_{0} <\infty \quad \text{ and } \qquad \int_{\R^{3}}|v|^{3}\,f_{0}(v)\d v < \infty.
\end{equation}
Actually, for simplicity, we shall assume that
$$N_0=2n_0, \qquad n_0 \geq 1.$$
\end{ass}

\subsection{Main results}

Our main result concerns the limit of the finite particle system, when the volume $\Lambda$, and the initial number of particles 
$N_0$ go to infinity in such a way that: 
\begin{eqnarray}
\label{S5thermod}
 \lim_{\Lambda, \,N_0 \to +\infty}\frac{N_0}{\Lambda}=1\in (0, +\infty).
\end{eqnarray}
As already observed, this limit has to be interpreted as a \textit{\textbf{thermodynamic limit}} (see \cite[Chapters 2 \& 3]{ruelle}). 
Setting for simplicity 
$$\varepsilon=\Lambda^{-1},$$ we introduce then the rescaled correlation functions
\begin{equation}\label{rescaledcorr}
f_{\ell}^{\varepsilon}(t,\Vk)=\sum_{N=\ell}^{\infty}\dfrac{\varepsilon^{\ell}}{(N-\ell)!}\int_{\R^{3(N-\ell)}} \P_{N}(t,\V)\d v_{\ell+1}\ldots\d v_{N}, \qquad \ell \geq 1.\end{equation}

Our first main result can be summarized as follows: 
\begin{theo}\label{main0}
 For any $\varepsilon >0$ and $N_{0} \in \N$ even, let $\{f_{\ell}^{\varepsilon}(t)\}_{\ell=1,\ldots,N_{0}}$ be the rescaled correlation functions associated to the unique solution $\{\P_{N}(t)\}_{N}$ to \eqref{PNVN} with initial datum \eqref{inDATA}. Then, for any $\ell \geq 1$ and any $t \geq 0$, there exists some positive measure $\bm{\mu}_{\ell}(t)\in \M(\R^{3\ell})$ and a subsequence (still denoted $\{f_{\ell}^{\varepsilon}(t)\}_{\varepsilon >0,N_{0} \in \N}$) such that
\begin{equation*}\lim_{\substack{\varepsilon \to 0,\,N_0 \to +\infty\\ \varepsilon\,{N_0}\to 1}}\big \langle f_{\ell}^{\varepsilon}(t), \Phi_{\ell}\big\rangle_{\ell}
=\big \langle \bm{\mu}_{\ell}(t), \Phi_{\ell} \rangle_{\ell} \qquad \forall \Phi_{\ell} \in \mathcal{C}_{0}(\R^{3\ell}).\end{equation*}
Moreover, the family $\{\bm{\mu}_{\ell}(\cdot)\}_{\ell \geq 1}$ is a weak solution to some suitable Annihilated Boltzmann Hierarchy of equations. 
Here, $\M(\R^{3\ell})$ denotes the space of signed Radon measures on $\R^{3\ell}$.
\end{theo}

We refer to Sections \ref{rescaled} -- \ref{sec:ABH} for the definition of the Annihilated Boltzmann Hierarchy (cf.~\eqref{sol:BHweak}) and a more precise statement, as well as the proof, of Theorem \ref{main0}(cf.~Propositions \ref{prp:conv} and Theorem \ref{main}). Moreover, we notice that the Annihilated Boltzmann Hierarchy admits as peculiar solution the tensorized function $\left\{f^{\otimes\ell}(\cdot)\right\}_{\ell}$, i.e. 
$$f^{\otimes\ell}(t,\bm{V}_{\ell})=f(t,v_{1})\ldots f(t,v_{\ell}), \qquad \quad \bm{V}_{\ell}=(v_{1},\ldots,v_{\ell}) \in \R^{3\ell}\;;\;t \geq 0,\;\ell \geq 1,$$
where $f(t)$ is the unique solution to \eqref{BE}.

We further notice that, strictly speaking, Theorem \ref{main0} does not provide a justification of \eqref{BE} since the measure solution $\bm{\mu}_{\ell}(t)$ may differ from the peculiar solution $f(t)^{\otimes\ell}.$ However, for bounded cross-sections, one can prove that the Annihilated Boltzmann Hierarchy admits a \textit{\textbf{unique}} solution, yielding a validation of \eqref{BE}, which is the content of our second main result:
\begin{theo}[\textit{\textbf{Propagation of Chaos}}]\label{main1}
Assume that $\Sigma_{\B}$ is bounded, i.e.~there exists $C_{\B} >0$ such that
$$\int_{\S} \B(v-v_{*},\omega)\d \omega \leq C_{\B} \qquad \forall v,v_{*} \in \R^{3}.$$
For any $\varepsilon >0$ and $N_{0} \in \N$ even, let $\{f_{\ell}^{\varepsilon}(t)\}_{\ell=1,\ldots,N_{0}}$ be the rescaled correlation functions associated to the unique solution $\{\P_{N}(t)\}_{N}$ to \eqref{PNVN} with initial datum \eqref{inDATA}. Then, for any $\ell \geq 1,$ we have
$$f_{\ell}^{\varepsilon}(t) \underset{N_{0}\varepsilon \to 1}{\underset{N_{0}\to \infty, \varepsilon \to 0}{\longrightarrow}} f(t)^{\otimes\ell}$$
in the weak-$\star$ topology of $\M(\R^{3\ell})$ and 
$f(t)$ is the unique solution to \eqref{BE}. The above convergence is uniform with respect to $t$ in any compact set.
\end{theo}

Though very partial, the above result is, up to our knowledge, the first rigorous mathematical justification of the kinetic equation \eqref{BE}. This result is restricted to bounded cross-sections and it would be of course more satisfactory to prove it for hard-sphere interactions. We explain in Section \ref{sec:HS} the main difficulty we face in proving such a result. We wish nevertheless to emphasize that the high dissipative nature of \eqref{BE} makes the analysis quite challenging. 

While the study of mean-field limit for spatially homogeneous equations reached already a mature level and the literature on this topic is very vast (as discussed in Section \ref{sec:kac}), the mathematical literature on \emph{dissipative kinetic-like equations} is rather scarse. We mention here the paper \cite{wen}, in the spirit of \cite{MM}, concerning the propagation of chaos for the inelastic Boltzmann equation (when the energy is dissipated but not the mass). In \cite{wen} suitable reservoirs are added to prevent too strong dissipation. 
The model with strong dissipation where mass is decreasing (and therefore, at particles level, the number of particles is non constant) which received more attention is the \emph{Smoluchowski equation} for coalescing particles. Probabilistic approaches to the justification of Smoluchowski  equation from a particle system have been discussed in \cite{HR}, \cite{norris} and \cite{wagner}. {For an analytical approach based on a suitable adaptation of the BBGKY Hierarchy, which is the inspiration for the present work, we refer to \cite{fede-esco}.} 
Notice here that, for the kinetic annihilation model considered here (cf.~\eqref{BE}--\eqref{PNVN}), probabilistic treatment does not seem easy to adapt since \eqref{BE} does not exhibit \emph{any} conserved quantity.  
We mention also some recent contributions on the derivation of a linear version of the Smoluchowski equation from a mechanical particle system (cf. \cite{NV, NNTV}) and the monograph \cite{kolok} for a comprehensive study of Smoluchowski or Boltzmann equation. 

As mentioned earlier, our approach is inspired by \cite{fede-esco} in particular for what concerns the use of suitable hierarchies of equations, the BBGKY Hierarchy and the Annihilated Boltzmann Hierarchy. The proof of Theorem \ref{main0} is obtained through a suitable compactness argument. The novelty of the approach, with respect to \cite{fede-esco}, is that we work here with measure solutions to the involved hierarchies. The main advantage of dealing with such a functional framework is that weak compactness is, somehow, equivalent to tightness and relies on some suitable estimates for the mass and the energy (uniformly with respect to the parameters $N,\varepsilon$). The proof of such weak compactness and partial convergence results are given in Section \ref{rescaled}, Proposition \ref{prp:conv}. Then, one can prove that any limiting point of the family of rescaled correlation functions $\{f_{\ell}^{\varepsilon}(t)\}_{\varepsilon >0,N_{0} \in \N}$ has to be measure solution to some suitable Annihilated Boltzmann Hierarchy introduced in Section \ref{sec:ABH}. As it is well known for this kind of problem, the major difficulty is to prove the uniqueness of solution to such limit hierarchy (see \cite{caprino} and \cite[Section 8]{MM}). We are able to prove such uniqueness result for bounded collision kernel, resulting in Theorem \ref{main1}. The extension of Theorem \ref{main1} to unbounded collision kernel -- mainly the hard-sphere collision kernel -- is the project of future work by the authors. The present contribution somehow paving the ground to the rigorous derivation of \eqref{BE}. We discuss in Section \ref{sec:HS} several possible paths to handle hard-sphere interactions and the technical difficulties associated to them. We can already anticipate that, for hard-sphere kernels, the exact rate of mass dissipation is not known and it appears difficult then to provide a suitable modification of \eqref{BE} where a suitable reservoir would prevent the convergence to zero of $f(t,v)$ (as in \cite{wen}). Another related difficulty is that it does not seem possible to adapt the approach introduced in \cite{caprino} and revisited in \cite[Section 8]{MM} and based upon De Finetti's Theorem since we are not dealing here with probability measures. We refer to Section \ref{sec:HS} for more details.

The organization of the paper is the following. In Section \ref{sec:many} we introduce the notations and function spaces we are dealing with in the sequel. We also prove the well-posedness of the $N$-particle system associated to \eqref{PNVN} owing to suitable substochastic semigroups techniques (see Theorem \ref{th:wellposPS}). In Section \ref{SMarginals} we define the (rescaled) correlation functions and the BBGKY hierarchy associated to \eqref{PNVN} while, in Section \ref{rescaled}, we establish our  main compactness result and the convergence part of Theorem \ref{main0}. In Section \ref{sec:ABH}, we derive the limiting hierarchy ABH and prove the second part of Theorem \ref{main0}, namely the fact that any limit point of the sequence of rescaled correlation functions is a suitable solution to this limiting hiearchy. We also prove Theorem \ref{main1} in Section \ref{sec:ABH} which relies on some uniqueness result for the hierarchy. Section \ref{sec:HS} describes the possible extension of the present results to the more relevant case of hard-sphere interactions, including several perspectives and discussions on the difficulties associated to this kind of interactions. In Appendix \ref{appendix1}, we prove a technical result concerning the semigroup generation properties used in Section \ref{sec:many}.

\section{{Well-posedness of the many-particles system \eqref{PNVN}}}\label{sec:many}
\subsection{Notations and function spaces.} For any $k \geq 2,$ unless otherwise specified, we will always assume that the functions $\Phi_{k}=\Phi(\bm{V}_{k})$, depending on $k$ velocities $v_{1},\ldots,v_{k}$, are symmetric with respect to permutations, i.e.
$$\Phi_{k}(v_{\pi(1)},\ldots,v_{\pi(k)})=\Phi_{k}(v_{1},\ldots,v_{k}) \qquad   \text{ for any permutation } \pi \in \mathfrak{S}_{k}$$
where $\mathfrak{S}_{k}$ is the permutation group of the set $\{1,\ldots,k\}$. For any $k \geq 1$, we denote by $\mathscr{M}(\R^{3k})$  the space of (bounded) signed Radon measures over $\R^{3k}$ endowed with the \emph{total variation} norm
$$\|\mu\|_{1,k}:=\int_{\R^{3k}}\left|\mu\right|(\d \bm{V}_{k}).$$
When no ambiguity can occur, we simply use $\|\mu\|_{1}$ for the above total variation norm. We see then $L^{1}(\R^{3k})$ as a (closed) subspace of $\mathscr{M}(\R^{3k})$ and still denote by $\|f\|_{1}$ the norm of a function $f \in L^{1}(\R^{3k})$. 
We refer to \cite{bogachev,evans} for   results and terminology about Radon measures. We notice that $\mathscr{M}(\R^{3k})$  -- endowed with the $\|\cdot\|_{1}$ norm is an $AL$-space (i.e a Banach lattice whose norm is additive on the positive cone), see \cite[Example 3, p. 218]{schaef}. The positive cone of $\mathscr{M}(\R^{3k})$ will be denoted by $\mathscr{M}^{+}(\R^{3k}).$ As we did for the functions $\Phi_{k}=\Phi(\bm{V}_{k})$, we also assume that the measures we consider in the sequel are symmetric with respect to permutations, i.e. $\mu \circ \pi^{-1}=\pi$  for any permutation $\pi \in \mathfrak{S}_{k}$.

For any $k \geq 1$, we denote by $\C_{b}(\R^{3k})$ (resp. $\C_{0}(\R^{3k})$) the space of continuous and bounded functions (resp. vanishing at infinity) over $\R^{3k}$ endowed with the sup-norm.

Given $k \geq 1$ and $\ell \in \{1,\ldots,k\}$, we define the marginal operator
$$\bm{\Pi}_{\ell}\::\:\mu_{k} \in \mathscr{M}(\R^{3k}) \longmapsto  \bm{\Pi}_{\ell}\,\mu_{k}:=\int_{\R^{3(k-\ell)}}\mu_{k}(\cdot,\d v_{\ell},\ldots,\d v_{k}) \in \mathscr{M}(\R^{3\ell})$$
i.e. $\bm{\Pi}_{\ell}\mu_{k}$ is defined in weak-form as 
$$\langle \bm{\Pi}_{\ell}\mu_{k}, \Phi_{\ell}\rangle_{\ell}=\int_{\R^{3k}}\Phi_{\ell}(\Vk)\mu_{k}(\d \bm{V}_{k}), \qquad \forall \Phi_{\ell} \in \C_{0}(\R^{3\ell}).$$
Given a measurable non-negative mapping $g_{k}\::\:\R^{3k}\to\R^{+}$ and $\mu_{k} \in \mathscr{M}(\R^{3k})$ we define the ``product'' $g_{k}(\cdot)\mu_{k}$ as the (possibly unbounded) Borel measure over $\R^{3k}$ defined by
$$\langle g_{k}(\cdot)\mu_{k},\Phi_{k}\rangle_{k}=\int_{\R^{3k}}g_{k}(\bm{V}_{k})\Phi_{k}(\bm{V}_{k})\mu_{k}(\d\bm{V}_{k}), \qquad \forall \Phi_{k} \in \C_{0}(\R^{3k}).$$

\subsection{Reformulation of the problem \eqref{PNVN}}\label{sec:reform}

It will be convenient to reformulate \eqref{PNVN} in a more compact form. 
We introduce, for any $N \geq 1$, the generalized collision frequency \begin{equation}\label{sigma}
\sigma_N \::\:\V=(v_1,\ldots,v_{N}) \in \R^{3N} \longmapsto {\sigma}_N(\V)=\sum_{1 \leq i < j \leq N} \int_{\S} \B(v_i-v_j,\omega)\d\omega
\end{equation}
and, introducing the functional space
$$L^{1}_{N}(\R^{3N})=\left\{f \in L^{1}(\R^{3N})\;;\;\|f\|_{L^1_{{N}}(\R^{3N})}:=\int_{\R^{3N}} f(\V)\sigma_N(\V)\d \V< \infty\right\}$$
we can define the following operators:
\begin{equation}\label{guadagno}
\begin{split}
\mathbf{G}_N\::\:L^1_N(\R^{3N}) &\longrightarrow L^1(\R^{3N})\\
 \Phi_N &\longmapsto \mathbf{G}_N (\Phi_N)(\V)= \sum_{1 \leq i < j \leq N}\int_{\S} \B
 (v_i-v_j,\omega) \Phi_N\left(\V^{i,j}\right)\d\omega \end{split}
 \end{equation}
 and  let $\mathbf{\Gamma}_{N+2}\::\:L^1_{N+2}(\R^{3N+6}) \longrightarrow L^1(\R^{3N})$ be defined by
\begin{equation}\label{op.gamma}
 \mathbf{\Gamma}_{N+2} (\Phi_{N+2})(\V)= \int_{\R^3}\d v_{N+1}\int_{\R^{3}}\d v_{N+2}\int_{\S} \B
(v_{N+1}-v_{N+2},\omega)\Phi_{N+2}\left(t,\V,v_{N+1},v_{N+2}\right)\d\omega
\end{equation}
for any $\Phi_{N+2} \in L^1_{N+2}(\R^{3N+6})$ and any $\V \in \R^{3N}$.
\bigskip

One can easily check that for any nonnegative $\Phi_N \in L^1_{{N}}(\R^{3N})$ and $\Phi_{N+2} \in L^1_{{N+2}}(\R^{3N+6})$ it holds:
\begin{itemize}
\item[i)] $\displaystyle \|\mathbf{G}_N[\Phi_N]\|_{L^1(\R^{3N})}=\|\Phi_N\|_{L^1_{{N}}(\R^{3N})}$; \vspace{1.5mm}
\item[ii)] $\displaystyle \|\mathbf{\Gamma}_{N+2}[\Phi_{N+2}]\|_{{L^1(\R^{3N})}} \leq \frac{2}{(N+1)(N+2)} \|\Phi_{N+2}\|_{{L^1_{N+2}(\R^{3N+6})}}.$
\end{itemize}
\medskip

With these notations, we can rewrite \eqref{PNVN} as
\begin{align}\label{PNop}
\partial_t \P_N(t,\V)+  {\frac{1}{\Lambda}}\sigma_N(\V)\P_N(t,\V)=& \, \frac{(1-\alpha)}{{\Lambda}}\mathbf{G}_N[\Psi_N](t,\V) \nonumber \\&
\,+ \frac{\alpha}{{2\Lambda}}\mathbf{\Gamma}_{N+2}[\P_{N+2}](t,\V) \qquad N \geq 1, \V \in \R^{3N}.
\end{align}
Let us introduce, for any $N \geq 1$, the (unbounded) operator
$$\mathcal{L}_N\::\:\D(\mathcal{L}_N) \subset L^1(\R^{3N}) \to L^1(\R^{3N})$$
with domain $\D(\mathcal{L}_N)=L^1_N(\R^{3N})$ and defined by
\begin{equation}\label{generatore}
\mathcal{L}_N (\Phi_N)=  (1-\alpha)\mathbf{G}_N[\Phi_N] - \sigma_N \Phi_N, \qquad \Phi_N \in \D(\mathcal{L}_N).
\end{equation}
The general properties of $\mathcal{L}_{N} $ are listed in the following Proposition whose proof is postponed to Appendix A:
\begin{prp}\label{propLN}
For any $N \geq 1$, $\left(\mathcal{L}_{N} ,\D(\mathcal{L}_{N} )\right)$ is the generator of a strongly continuous semigroup of contractions 
$\left(\mathcal{S}_N(t)\right)_{t \geq 0}$ in $L^1(\R^{3N})$. Moreover, for  any $t > 0$, one has \begin{equation}\label{smoothing}
\int_{0}^t \left\|\mathcal{S}_{N} (s)\Phi_{N}\right\|_{L^{1}_{N}(\R^{3N})} \d s \leq \dfrac{1}{\alpha} \|\Phi_{N}\|_{L^{1}(\R^{3N})}\qquad \forall \Phi_{N} \in L^{1}(\R^{3N}).\end{equation}
\end{prp}
\medskip

Using \eqref{generatore} we can rewrite \eqref{PNop} as
\begin{equation} \label{PNop2}
\partial_t \P_N(t,\V)=\frac{1}{\Lambda}\mathcal{L}_{N}[\Psi_N](t,\V) + 
\frac{\alpha}{{2\Lambda}}\mathbf{\Gamma}_{N+2}[\P_{N+2}](t,\V).
\end{equation}
Thanks to Proposition \ref{propLN}, seeing the factor $\frac{1}{\Lambda}$ as a time scaling, we can write the solution to \eqref{PNop2} with initial datum $\P_N(0)$ in terms of the Duhamel's formula
\begin{equation}\label{duha}
\P_N(t)=\mathcal{S}_N\left(\frac{t}{\Lambda}\right)\P_N(0)+\frac{\alpha}{{2\Lambda}} \int_0^{t} \mathcal{S}_N\left(\frac{t-s}{\Lambda}\right)\mathbf{\Gamma}_{N+2}[\P_{N+2}(s)]\d s.
\end{equation}

\subsection{A priori estimates}

We establish here several \emph{a priori} estimates for the solution to \eqref{PNVN}. 
We consider a sequence $\{\P_N(t)\}_{N}$ of nonnegative functions with
$$\P_{N} \in \mathcal{C}^{1}([0,+\infty); L^{1}(\R^{3N})) \qquad \forall N \geq 1$$
which satisfies \eqref{PNVN} for any $N \geq 1$. 
We introduce then, for every $N\in \mathbb{N}$, the function
\begin{eqnarray}
\label{S2E4}
\mathbf{\P}(t, N)= \frac{1}{N!}\int_{\R^{3N}} \P_N(t,\V)\d \V
\end{eqnarray}
which represents the probability that at time $t$ the system is constituted by $N$ particles. We 
observe that the natural normalization (\ref{S2E2}) holds true for any time, i.e. 
\begin{eqnarray}
\label{conservazione}
\sum_{N=0}^{\infty}\mathbf{\P}(t, N)= 1
\end{eqnarray}
where we notice that the factor $1/N!$ in definition (\ref{S2E4}) is needed to compensate for counting all the $N!$ physically equivalent ways of arranging the velocities $v_1,\dots,v_{N}$. 

 Due to our choice of the initial datum we have
\begin{equation}
\label{S2E5}
\mathbf{\P}(0,N)=:\mathbf{\P}^0(N)=\delta_{N_0}(N),
\end{equation}
where $\delta_{m}$ is the Dirac counting measure on $\N$.
Equation \eqref{S2E5} just means that at $t=0$ our system has exactly  $N_0$ particles and it follows that at $t=0$ property \eqref{conservazione} is trivially satisfied. We now  show that \eqref{conservazione}  holds for any positive $t>0$ as well:
\begin{prp}\label{prop:apriori}
Assume that the mapping $t \geq 0 \mapsto \sum_{N=1}^{\infty}\P_{N}(t)$ belongs to $\mathcal{C}^{1}([0,\infty))$. Then, for any $t \geq 0$, it holds
\begin{equation}\label{support}\P_{N}(t)=0 \qquad \forall N \geq N_{0}+1\end{equation}
and 
\begin{equation}\label{mass}\sum_{N=1}^{\infty}\mathbf{\P}(t,N)=1.\end{equation}
\end{prp}
\begin{proof} The proof is an adaptation of the argument proposed in  \cite[Theorem 2.1]{fede-esco}. More precisely, for any $k \geq 1$, we compute the time derivative of $\sum_{N=k}^{\infty}\mathbf{\P}(t,N)$. By \eqref{PNVN}  we get 
\begin{eqnarray}\label{eq:CONS}
&&\!\!\frac{\d }{\d t}\sum_{N=k}^\infty \mathbf{\P}(t, N)=-\frac12 \sum_{N=k}^{\infty}\frac{N(N-1)}{N!}\int_{\R^{3N}} \d \V\int_{\S} \B_{{\Lambda}}(v_{N-1}-v_{N},\omega) \P_N\left(t,\V\right)\d \omega\ \ \ \ \nonumber\\
&&\qquad\qquad+\frac{(1-\alpha)}{2}\sum_{N=k}^{\infty}\frac{N(N-1)}{N!}\int_{\R^{3N}} \d \V\int_{\S} \B_{{\Lambda}}(v_{N-1}-v_{N},\omega) \P_N\left(t,\V^{N-1, N}\right)\d \omega\nonumber\\
&&\qquad\qquad+\frac{\alpha}{2} \sum_{N=k}^{\infty}\frac{1}{N!}\int_{\R^{3N+6}} \d \bm{V}_{N+2}\int_{\S} \B_{{\Lambda}}(v_{N+1}-v_{N+2},\omega) \P_{N+2}(t,\bm{V}_{N+2})\d \omega
,\nonumber\\
&&
\end{eqnarray}
where we used the symmetry of $\P_N(t)$ with respect to any permutation of the indices to write
\begin{multline*}
\sum_{1 \leq i < j \leq N}\int_{\R^{3N}}  \d \V\int_{\S} \B_{{\Lambda}}(v_i-v_j,\omega)\P_N\left(t,\V\right)\d\omega=\\
\frac{N(N-1)}{2}\int_{\R^{3N}}  \d \V\int_{\S} \B_{{\Lambda}}(v_{N-1}-v_{N},\omega) \P_N\left(t,\V\right)\d \omega\end{multline*}
and
\begin{multline*}\sum_{1 \leq i < j \leq N}\int_{\R^{3N}} \d \V\int_{\S} \B_{{\Lambda}}(v_i-v_j,\omega)\P_N\left(t,\V^{i,j}\right)\d\omega=\\
\frac{N(N-1)}{2}\int_{\R^{3N}} \d \V\int_{\S} \B_{{\Lambda}}(v_{N-1}-v_{N},\omega) \P_N\left(t,\V^{N-1, N}\right)\d \omega.
\end{multline*}
\medskip

\noindent We observe that for any $i,j\in \{1,\dots, N\}$ we have
\begin{equation*}
\int_{\R^{3}} \d v_{i}\int_{\R^{3}} \d v_j\int_{\S} \B_{{\Lambda}}(v_i-v_j,\omega) \P_N\left(t,\V\right)\d \omega=\int_{\R^{3}} \d v_i\int_{\R^{3}} \d v_j\int_{\S} \B_{{\Lambda}}(v_i-v_j,\omega) \P_N\left(t,\V^{i,j}\right)\d \omega.
\end{equation*}
Then, equation (\ref{eq:CONS}) yields
\begin{equation*}\begin{split}\label{eq:CONS1}
\frac{\d }{\d t}\sum_{N=k}^\infty \mathbf{\P}(t, N)&=-\frac{\alpha}{2}\sum_{N=k}^{\infty}\frac{N(N-1)}{N!}\int_{\R^{3N}} \d \V\int_{\S} \B_{{\Lambda}}(v_{N-1}-v_{N},\omega) \P_N\left(t,\V\right)\d \omega\\
&\quad +\frac{\alpha}{2} \sum_{N=k}^{\infty}\frac{1}{N!}\int_{\R^{3N+6}} \d \bm{V}_{N+2}\int_{\S} \B_{{\Lambda}}(v_{N+1}-v_{N+2},\omega) \P_{N+2}(t,\bm{V}_{N+2})\d \omega,\\
&=-\frac{\alpha}{2}\sum_{N=k+2}^{\infty}\frac{1}{(N-2)!}\int_{\R^{3N}} \d \V\int_{\S} \B_{{\Lambda}}(v_{N-1}-v_{N},\omega) \P_N\left(t,\V\right)\d \omega\\
&\quad +\frac{\alpha}{2} \sum_{N=k}^{\infty}\frac{1}{N!}\int_{\R^{3N+6}} \d \bm{V}_{N+2}\int_{\S} \B_{{\Lambda}}(v_{N+1}-v_{N+2},\omega) \P_{N+2}(t,\bm{V}_{N+2})\d \omega
\end{split}\end{equation*}
Performing the change of indices $N \to N-2$ in the first sum, one finds
\begin{equation}\begin{split}\label{eq:CONS1}
\frac{\d }{\d t}\sum_{N=k}^\infty \mathbf{\P}(t, N)&=-\frac{\alpha}{2}\dfrac{1}{(k-2)!}\int_{\R^{3k}}\d \bm{V}_{\ell}\int_{\S}\B_{\Lambda}(v_{\ell}-v_{k-1},\omega)\P_{\ell}(t,\bm{V}_{\ell})\d\omega\\
&\qquad -\frac{\alpha}{2}\dfrac{1}{(k-1)!}\int_{\R^{3k+3}}\d \bm{V}_{\ell+1}\int_{\S}\B_{\Lambda}(v_{\ell+1}-v_{\ell},\omega)\P_{\ell+1}(t,\bm{V}_{\ell+1})\d\omega\leq 0.
\end{split}\end{equation}
In particular, applying \eqref{eq:CONS1} and choosing $k=N_{0}+1$, we obtain
$$\sum_{N=N_{0}+1}^{\infty}\mathbf{\P}(t,N) \leq \sum_{N=N_{0}+1}^{\infty}\mathbf{\P}(0,N)=0 \qquad \forall t \geq 0$$
since $\mathbf{\P}^{0}(N) =0$ for any $N \neq N_{0}.$ Using that $\mathbf{\P}(t,N)$ is nonnegative for any $t \geq 0$ and any $N \geq 1$, we get \eqref{support}. The proof of \eqref{mass} follows directly from \eqref{eq:CONS1} applied to $k=1$ which yields
$$\dfrac{\d}{\d t}\sum_{{N=1}}^{\infty}\mathbf{\P}(t,N)=0$$
and the conclusion follows since $\sum_{N=1}^{\infty}\mathbf{\P}^{0}(N)=1.$
\end{proof}

We can now prove the well-posedness of the Cauchy problem associated to \eqref{PNVN} (cf. \eqref{PNop2}).  We establish the following:
\begin{theo}\label{th:wellposPS}
For any $n_{0}\geq 1$ let $N_{0}=2n_{0}$ and let the initial datum $\left(\P_N(0)\right)_{N \geq 1}$ be given by \eqref{inDATA}. Then, there exists a unique solution $\left\{\P_{N}(t)\right\}_{N \geq 1,\;t \geq 0}$ to \eqref{PNVN} such that
$$\P_{N} \in \mathcal{C}^{1}([0,+\infty); L^{1}(\R^{3N})) \qquad \forall N \geq 1$$
and such that the mapping $t \geq 0 \mapsto \sum_{N=1}^{\infty}\mathbf{\P}(t,N)$ belongs to $\mathcal{C}^{1}([0,\infty))$. Such a solution satisfies \eqref{support} and \eqref{mass} and is given by
\begin{equation}\label{N<N_0}
\begin{cases}
\P_{N_0}(t)&=\mS_{N_0}\left(\frac{t}{\Lambda}\right)\P_{N_0}(0) \quad t \geq 0\\
\P_{N_0-1}(t)&=0 \qquad t \geq 0\\
\P_{N_0-2}(t)&=\displaystyle\frac{\alpha}{{\Lambda}} \dis \int_0^t \mS_{N_0-2}\left(\frac{t-s}{\Lambda}\right) \mathbf{\Gamma}_{N_0}[\Psi_{N_0}(s)]\d s\\
&=\displaystyle\frac{\alpha}{{\Lambda}} \dis \int_0^t \mS_{N_0-2}\left(\frac{t-s}{\Lambda}\right)\mathbf{\Gamma}_{N_0}[\mS_{N_0}\left(\frac{s}{\Lambda}\right)\P_{N_0}(0)]\d s\\
\P_{N_0-3}(t)&=0\\
\P_{N_0-4}(t)&=\displaystyle\frac{\alpha}{{\Lambda}} \dis \int_0^t \mS_{N_0-4}\left(\frac{t-s}{\Lambda}\right)\mathbf{\Gamma}_{N_0-2}[\Psi_{N_0-2}(s)]\d s\\
&= \displaystyle\frac{\alpha^2}{{\Lambda^2}}\dis \int_0^t \d s \int_0^s  \mS_{N_0-4}\left(\frac{t}{\Lambda}\right)\mathbf{\Gamma}_{N_0-2}\left[\mS_{N_0-2}\left(\frac{s-\tau}{\Lambda}\right)\mathbf{\Gamma}_{N_0}[\mS_{N_0}\left(\frac{\tau}{\Lambda}\right)\P_{N_0}(0)]\right]\d\tau\\
&\vdots\,\\
\P_{N_0-2k}(t)&=\displaystyle\frac{\alpha}{{\Lambda}} \dis \int_0^t \mS_{N_0-2k}\left(\frac{t-s}{\Lambda}\right)\mathbf{\Gamma}_{N_0-2k+2}[\Psi_{N_0-2k+2}(s)]\d s \qquad \forall k \in \{1,\ldots,n_{0}\}.
\end{cases}
\end{equation}
\medskip

Moreover, for any $T > 0$, it holds
\begin{equation}\label{L10T}\int_{0}^{T}\|\P_{N}(t)\|_{L^{1}_{N}(\R^{3N})}\d t \leq \frac{\Lambda}{\alpha}\,N_0!\, \qquad \forall N=N_{0}-2k\,,\; k \in \{0,\ldots,n_{0}\}.\end{equation}
\end{theo}
\begin{proof}
We introduce the class $\mathscr{S}$ of sequences $\left\{\Phi_{N}(\cdot)\right\}_{N \geq 1}$ such that 
$$\Phi_{N} \in \mathcal{C}^{1}([0,+\infty); L^{1}(\R^{3N})) \qquad N \geq 1$$ and the mapping $t \geq 0 \longmapsto \sum_{N=1}^{\infty}\overline{\Phi}(t,N)$ belongs to $\mathcal{C}^{1}([0,\infty))$ where
$$\overline{\Phi}(t,N)=\dfrac{1}{N!}\int_{\R^{3N}}\Phi_{N}(t,\V)\d \V \qquad \forall t \geq 0.$$ 
According to Proposition \ref{prop:apriori}, any solution $\left\{\P_{N}(t)\right\}_{N \geq 1,\,t \geq 0}$ to \eqref{PNVN} (with initial datum given by \eqref{inDATA}) which belongs to the class $\mathscr{S}$ is such that $\P_{N}(t)=0$ for any $N \geq N_{0}.$ 
Then, Duhamel's formula \eqref{duha} implies that any solution $\P_{N}(t)$ satisfies \eqref{N<N_0}. It remains to prove that, indeed, the sequence $\left\{\P_{N}(t)\right\}_{N \geq 1,\;t \geq 0}$ given by \eqref{N<N_0} belongs to the class $\mathscr{S}$ and that the solution is unique within this class. To prove the first claim, we first  prove that $\P_{N}(t) \in L^{1}(\R^{3N})$ for any $t \geq 0$ and any $N \in \mathbb{N}.$ Clearly, the difficulty stems from the fact that $\sigma_{N}$ is unbounded. Let us fix $T > 0$. One has 
\begin{equation}\label{PN0}\begin{split}
\int_{0}^{T}\|\P_{N_{0}}(t)\|_{L^{1}_{N_{0}}}\d t&=\int_{0}^{T}\left\|\mS_{N_{0}}\left(\frac{t}{\Lambda}\right)\P_{N_{0}}(0)\right\|_{L^{1}_{N_{0}}}\d t\\
&\leq \frac{\Lambda}{\alpha}\|\P_{N_{0}}\|_{L^{1}(\R^{3N_{0}})}=\dfrac{N_{0}!\,\Lambda}{\alpha}\end{split}\end{equation}
according to \eqref{smoothing}. Now, for any $k \in \{1,\ldots,n_{0}\}$, using \eqref{N<N_0}, we have
\begin{multline*}
\int_{0}^{T}\|\P_{N_0-2k}(t)\|_{L^{1}_{N_0-2k}}\d t \\
\leq \frac{\alpha}{\Lambda}\int_{0}^{T}\d t\int_0^t \left\|\mS_{N_0-2k}\left(\frac{t-s}{\Lambda}\right)\mathbf{\Gamma}_{N_0-2k+2}[\Psi_{N_0-2k+2}(s)]\right\|_{L^{1}_{N_{0}-2k}}\d s\\
=\frac{\alpha}{\Lambda}\int_{0}^{T}\d s\int_{s}^{T} \left\|\mS_{N_0-2k}\left(\frac{t-s}{\Lambda}\right)\mathbf{\Gamma}_{N_0-2k+2}[\Psi_{N_0-2k+2}(s)]\right\|_{L^{1}_{N_{0}-2k}}\d t\\
\leq  \frac{\alpha}{\Lambda}\int_{0}^{T}\d s\int_{0}^{T} \left\|\mS_{N_0-2k}\left(\frac{\tau}{\Lambda}\right)\mathbf{\Gamma}_{N_0-2k+2}[\Psi_{N_0-2k+2}(s)]\right\|_{L^{1}_{N_{0}-2k}}\d \tau.\end{multline*}
Applying again \eqref{smoothing} we obtain
\begin{multline*}
\int_{0}^{T}\|\P_{N_0-2k}(t)\|_{L^{1}_{N_0-2k}}\d t \leq  \int_{0}^{T}\left\|\mathbf{\Gamma}_{N_{0}-2k+2}[\Psi_{N_{0}-2k+2}(s)]\right\|_{L^{1}(\R^{3(N_{0}-2k+2)}}\d t\\
\leq  \int_{0}^{T}\|\Psi_{N_{0}-2k+2}(s)\|_{L^{1}_{N_{0}-2k+2}}\d t \qquad \forall k \in \{1,\ldots,n_{0}\}.
\end{multline*}
By finite induction and, using \eqref{PN0}, this clearly yields \eqref{L10T}. In particular, using the fact that $(\mS_{N}(t))_{t \geq 0}$ is a contraction semigroup in $L^{1}(\R^{3N})$,  one deduces directly from \eqref{L10T} that
$$\|\P_{N+2}(t)\|_{L^{1}} \leq \int_{0}^{t}\left\|\mathbf{\Gamma}_{N}[\Psi_{N}(s)]\right\|_{L^{1}_{N}}\d s \leq C_{N} \qquad \forall N \leq N_{0}$$ which shows that $\P_{N}(t)\in L^{1}_{N}(\R^{3N})$ for any $t \geq 0$. Moreover, by virtue of \eqref{normalization}, the initial datum $\P_{N_{0}}^{0} \in \D(\mathcal{L}_{N_{0}})=L^{1}(\R^{3N_{0}}) \cap L^{1}_{N_{0}}(\R^{3N_{0}}$. Thus, by classical semigroup theory $\P_{N_{0}}(t)=\mS_{N_{0}}(t)\P_{N_{0}}^{0} \in \mathcal{C}^{1}([0,\infty),L^{1}(\R^{3N_{0}}).$ By \eqref{PNVN}, it follows that $\P_{N} \in \mathcal{C}^{1}([0,+\infty); L^{1}(\R^{3N}))$ for any $N \leq N_{0}$. Moreover, since the sum $\sum_{N \geq 1}\mathbf{\P}(\cdot,N)$ is actually finite, it also belongs to $\mathcal{C}^{1}([0,+\infty)).$ This shows that the constructed solution $\left\{\P_{N}(t)\right\}_{N \geq 1,\;t \geq 0}$ belongs to the class $\mathscr{S}.$ It remains to prove the uniqueness of the solution to \eqref{PNVN} within this class. In order to do this, let us consider two solutions  $\left\{\P^{1}_{N}(t)\right\}_{N \geq 1,\,t \geq 0}$ and  $\left\{\P_{N}^{2}(t)\right\}_{N \geq 1,\,t \geq 0}$ of \eqref{PNVN}, both belonging to the class $\mathscr{S}$ and such that $\P_{N}^{1}(0)=\P^{2}_{N}(0)=\P_{N}^{0}$ given by \eqref{inDATA}. We set $\Phi_{N}(t)=\P^{1}_{N}(t)-\P^{2}_{N}(t)$ for any $t \geq 0$ and any $N \geq 0$. Arguing exactly as above one deduces from Duhamel's formula \eqref{duha} that 
$$\int_{0}^{T}\|\Phi_{N}(t)\|_{L^{1}_{N}}\d t \leq \dfrac{1}{\Lambda}\int_{0}^{T}\|\Phi_{N+2}(t)\|_{L^{1}_{N+2}}\d t \qquad \forall \,T > 0.$$
Since  \eqref{support} implies that $\Phi_{N}(t)=0$ for any $N \geq N_{0}+1$, this shows that $\Phi_{N}(t)=0$ for any $N \geq 0$ and the solution is unique {in} the class $\mathscr{S}.$\end{proof}

\section{Correlation Functions and BBGKY hierarchy}\label{SMarginals}
\setcounter{equation}{0}
\setcounter{theo}{0}
For any fixed $\ell\in \N$, we define the $\ell$-particle correlation function $f_\ell(v_1, \dots, v_\ell,t)$ 
at time $t$ as
\begin{eqnarray}\label{MARGINAL}
f_\ell(v_1,\dots, v_\ell,t) = \sum_{N=\ell}^\infty
 \frac{1}{(N-\ell)!}\int_{\R^3} \d v_{\ell+1}\dots \int_{\R^3}  \d v_{N}\P_{N}(\V,t) .
\end{eqnarray}
For $\ell=0$ we set $f_0\equiv 0$. Notice that, for $N=\ell$, the above integral is meaningless and it is intended simply as $\P_{\ell}(\Vk,t),$ i.e.
$$f_{\ell}(\Vk,t)=\P_{\ell}(\Vk,t)+\sum_{N=\ell+1}^\infty
 \frac{1}{(N-\ell)!}\int_{\R^3} \d v_{\ell+1}\dots \int_{\R^3}  \d v_{N}\  \P_{N}(\V,t).$$
We collect here some general properties of the correlation functions $f_\ell$. At any time $t$ the expected (or mean) number of particles $\overline {N}(t)$ defined as
\begin{equation}
\overline{N(t)}=\sum_{N=0}^\infty
 N \mathbf{\P}(t, N),
\end{equation}
with $\mathbf{\P}(t, N)$ given as in \eqref{S2E4}, satisfies
\begin{eqnarray}
\label{S3E100}
\overline{N(t)}=\int_{\R^{3}} f_1(v_1, t)\d v_1.
\end{eqnarray}
The function $f_1$ is then the density function associated to the average number of particles. More generally, one may also define
\begin{equation}
\overline{N_k(t)}=\sum_{N=1}^\infty 
N(N-1)\dots (N-k+1) \mathbf{\P}(t, N),
\end{equation}
and then
\begin{eqnarray}
\label{S4DefMj}
\overline{N_k(t)}=\|f_k(t)\|_{L^1(\R^{3k})}.
\end{eqnarray}

The functions $f_\ell$'s will be called  correlation functions since they satisfy properties  (\ref{S3E100}), (\ref{S4DefMj}) and  their definition is similar  to that of the classical correlation functions in statistical mechanics. 
The correlation function $f_{\ell}(t)=f_{\ell}(t,\Vk)$ satisfies the following 
\begin{prp}[\textit{\textbf{BBGKY Hierarchy}}]\label{prp:bbgky}
Let $f_{\ell}(t)=f_{\ell}(t,\Vk)$ be defined as in \eqref{MARGINAL}. For any $\ell \in \{1,\ldots,N_{0}\}$ we have that $f_{\ell}(t,\Vk)$ satisfies 
\begin{eqnarray}\label{hierarchyNOTscaledb}
&&\partial_t f_\ell(t,\Vk)=\sum_{1 \leq i < j \leq \ell}\int_{\S} \B_{\Lambda}(v_i-v_j,\omega) \left[(1-\alpha)f_\ell\left(t,\Vk^{i,j}\right)-f_\ell\left(t,\Vk\right)\right]\d\omega\nonumber\\
&&\qquad \ \ +\sum_{i=1}^{\ell}\int_{\R^{3}} \d	v_{\ell+1}\int_{\S} \B_{\Lambda}(v_i-v_{\ell+1},\omega) \left[(1-\alpha)f_{\ell+1}\left(t,\Vku^{i, \ell+1}\right)-f_{\ell+1}\left(t,\Vku\right)\right]\d\omega.\nonumber\\
&&
\end{eqnarray}
\end{prp}
\begin{proof} The proof is obtained by direct inspection, exploiting the fact that $\P_{N}(t)$ is symmetric. Namely, using the definition (\ref{MARGINAL}), a straightforward computation from (\ref{PNop}) shows that for any $\ell\le N_0$ the function $f_{\ell}$ satisfies
\begin{align}\label{BBGKY}
\partial_t f_\ell (\Vk, t)= & -\sum_{N=\ell}^\infty \frac{1}{(N-\ell)!}\sum_{1 \leq i < j \leq N}\int_{\R^{3(N-\ell)}} \d  \bm{V}_{N, \ell}\int_{\S} \B_{\Lambda}(v_i-v_j,\omega) \P_N\left(t,\V\right)\d\omega \nonumber \\&
+(1-\alpha) \sum_{N=\ell}^{\infty} \frac{1}{(N-\ell)!}\sum_{1 \leq i < j \leq N}\int_{\R^{3(N-\ell)}} \d  \bm{V}_{N, \ell}\int_{\S} \B_{\Lambda}(v_i-v_j,\omega) \P_N\left(t,\V^{i,j}\right)\d\omega \nonumber \\&
+\frac{\alpha}{2}\sum_{N=\ell}^{\infty} \frac{1}{(N-\ell)!}\int_{\R^{3(N+2-\ell)}} \d  \bm{V}_{N+2, \ell}\int_{\S}\B_{\Lambda}
(v_{N+1}-v_{N+2},\omega)\P_{N+2}\left(t,\bm{V}_{N+2}\right)\d\omega \nonumber \\&
:=A_{1}+A_{2}+A_{3}
\end{align}
where we used the notation, valid for any $N \geq 1$, $k \in \{1,\ldots,N\}$, 
$$\bm{V}_{N, k}:= (v_{k+1},\dots,v_{N}).$$
The first term on the right hand side of (\ref{BBGKY}) gives the following contributions
\begin{align}\label{BBGKYgain}
A_{1}=& -\sum_{1 \leq i < j \leq \ell}\int_{\S} \B_{\Lambda}(v_i-v_j,\omega) f_\ell\left(t,\Vk\right)\d\omega \nonumber \\&
-\sum_{N=\ell}^{\infty} \frac{(N-\ell)}{(N-\ell)!}\sum_{ i =1}^{ \ell}\int_{\R^{3}} \d v_{\ell+1}\int_{\S} \B_{\Lambda}(v_i-v_{\ell+1},\omega)\d\omega \int_{\R^{3(N-\ell-1)}} \P_N\left(t,\V\right)\d  \bm{V}_{N, \ell+1} \nonumber \\&
-\frac{1}{2}\sum_{N=\ell}^{\infty} \frac{(N-\ell)(N-\ell-1)}{(N-\ell)!} \int_{\R^{3(N-\ell)}}\d  \bm{V}_{N+2, \ell+2}\int_{\S} \B_{\Lambda}(v_{\ell+1}-v_{\ell+2},\omega)\P_N\left(t,\V\right)\d\omega \nonumber \\&
=:A_{1,1}+A_{1,2}+A_{1,3}
\end{align}
where we divided the sum with respect to $i$ and $j$ into three parts, $1\leq i<j\leq \ell$, $1\leq i\leq \ell,\ \ell+1\leq j\leq N$ and $\ell+1\leq i<j\leq N$. Moreover, we used the symmetry of $\P_N(t)$ with respect to any permutation of the indices  
to write
\begin{eqnarray*}
&& \sum_{j=\ell+1}^N \int_{\R^{3(N-\ell)}} \d  \bm{V}_{N, \ell}\int_{\S} \B_{\Lambda}(v_i-v_{j},\omega)\P_N\left(t,\V\right)\d\omega
\nonumber\\
&&=(N-\ell)  \int_{\R^{3(N-\ell)}} \d  \bm{V}_{N, \ell}\int_{\S} \B_{\Lambda}(v_i-v_{\ell+1},\omega)\P_N\left(t,\V\right)\d\omega\nonumber\\
&&=(N-\ell)\int_{\R^{3}} \d v_{\ell+1} \int_{\S} \B_{\Lambda}(v_i-v_{\ell+1},\omega) \int_{\R^{3(N-\ell-1)}} \d\bm{V}_{N, \ell+1}\P_N\left(t,\V\right)\d\omega
\end{eqnarray*}
and
\begin{eqnarray*}
&& \sum_{\ell+1\leq i<j\leq N} \int_{\R^{3(N-\ell)}} \d  \bm{V}_{N, \ell} \int_{\S} \B_{\Lambda}(v_i-v_{j},\omega)\P_N\left(t,\V\right)\d\omega
\nonumber\\
&&=\dfrac{(N-\ell)(N-\ell-1)}{2}  \int_{\R^{3(N-\ell)}} \d  \bm{V}_{N, \ell}\int_{\S} \B_{\Lambda}(v_{\ell+1}-v_{\ell+2},\omega)\P_N\left(t,\V\right)\d\omega\nonumber\\
&&=\dfrac{(N-\ell)(N-\ell-1)}{2} \int_{\R^{3}} \d v_{\ell+1}\int_{\R^{3}} \d v_{\ell+2}\int_{\S} \B_{\Lambda}(v_{\ell+1}-v_{\ell+2},\omega)\d\omega\times\\
&&\phantom{+++++++++} \times\int_{\R^{3(N-\ell-2)}}\P_N\left(t,\V\right) \d  \bm{V}_{N, \ell+2}.
\end{eqnarray*}
Using Fubini's theorem and Definition (\ref{MARGINAL}) we obtain
\begin{equation}\begin{split}\label{BBGKYgain1}
&A_{1,2}=- \sum_{N=\ell+1}^{\infty} \frac{(N-\ell)}{(N-\ell)!}\sum_{ i =1}^{\ell}\int_{\R^{3}} \d v_{\ell+1}\int_{\S} \B_{\Lambda}(v_i-v_{\ell+1},\omega)\d\omega\int_{\R^{3(N-\ell-1)}} \P_N\left(t,\V\right)\d  \bm{V}_{N, \ell+1} \\
&=-\sum_{ i =1}^{\ell}\int_{\R^{3}} \d v_{\ell+1}\int_{\S} \B_{\Lambda}(v_i-v_{\ell+1},\omega)\d\omega
\sum_{N=\ell+1}^{\infty} \frac{1}{\left(N-(\ell+1)\right)!}\int_{\R^{3(N-\ell-1)}} \P_N\left(t,\V\right)\d  \bm{V}_{N, \ell+1}\\
&= -\sum_{ i =1}^{\ell}\int_{\R^{3}} \d v_{\ell+1}\int_{\S} \B_{\Lambda}(v_i-v_{\ell+1},\omega)f_{\ell+1}\left(t,\Vku\right)\d\omega
\end{split}\end{equation}
and
\begin{align}\label{BBGKYlossK+2}
A_{1,3}&=-\frac{1}{2}\sum_{N=\ell+2}^{\infty} \frac{(N-\ell)(N-\ell-1)}{(N-\ell)!}\int_{\R^{3}} \d v_{\ell+1}\int_{\R^{3}} \d v_{\ell+2} \nonumber \\
&\phantom{++++}\int_{\S} \B_{\Lambda}(v_{\ell+1}-v_{\ell+2},\omega)\d\omega\int_{\R^{3(N-\ell-2)}} \P_N\left(t,\V\right)\d \bm{V}_{N, \ell+2} \nonumber \\
&=-\frac{1}{2}\int_{\R^{3}} \d v_{\ell+1}\int_{\R^{3}} \d v_{\ell+2}\int_{\S} \B_{\Lambda}(v_{\ell+1}-v_{\ell+2},\omega)
f_{\ell+2}\left(t,\bm{V}_{\ell+2}\right)\d\omega.
\end{align}
Therefore, using \eqref{BBGKYgain1} and \eqref{BBGKYlossK+2}, we obtain that $A_{1}$, i.e.~the first term on the right hand side of (\ref{BBGKY}), 
 gives the following contribution
\begin{align}\label{BBGKYloss}
A_{1}=& -\sum_{1 \leq i < j \leq \ell}\int_{\S} \B_{\Lambda}(v_i-v_j,\omega) f_\ell\left(t,\Vk\right)\d\omega \nonumber \\&
-\sum_{ i =1}^{\ell}\int_{\R^{3}} \d v_{\ell+1}\int_{\S} \B_{\Lambda}(v_i-v_{\ell+1},\omega)f_{\ell+1}\left(t,\Vku\right)\d\omega \nonumber \\&
-\frac{1}{2}\int_{\R^{3}} \d v_{\ell+1}\int_{\R^{3}} \d v_{\ell+2}\int_{\S} \B_{\Lambda}(v_{\ell+1}-v_{\ell+2},\omega)f_{\ell+2}\left(t,\bm{V}_{\ell+2}\right)\d\omega.
\end{align}
By analogous computations we get that $A_{2}$, i.e.~the second term on the right hand side of (\ref{BBGKY}), yields
\begin{align}\label{gainALPHAN}
A_{2}=&\, (1-\alpha) \sum_{1 \leq i < j \leq \ell}\int_{\S} \B_{\Lambda}(v_i-v_j,\omega) f_\ell\left(t,\Vk^{ij}\right)\d\omega \nonumber \\&
+(1-\alpha) \sum_{ i =1}^{\ell}\int_{\R^{3}} \d v_{\ell+1}\int_{\S} \B_{\Lambda}(v_i-v_{\ell+1},\omega)f_{\ell+1}\left(t,\Vku^{i, \ell+1}\right)\d\omega \nonumber \\&
+\frac{(1-\alpha)}{2}\int_{\R^{3}} \d v_{\ell+1}\int_{\R^{3}} \d v_{\ell+2}\int_{\S} \B_{\Lambda}(v_{\ell+1}-v_{\ell+2},\omega)f_{\ell+2}\left(t,\bm{V}_{\ell+2}^{\ell+1, \ell+2}\right)\d\omega.
\end{align}
We now look at the third term on the right hand side of (\ref{BBGKY}), i.e.~$A_{3}$. 
Due to the symmetry of $\P_{N}(t)$ with respect to any permutation of the indeces it holds
\begin{align}\label{gainN+2b}
A_{3} &= \, \frac{\alpha}{2}\sum_{N=\ell}^{\infty} \frac{1}{(N-\ell)!}\int_{\R^{3(N-\ell)}} \d  \bm{V}_{N, \ell}\int_{\R^{3}} \d v_{N+1}\int_{\R^{3}}\d v_{N+2}\nonumber \\&
\quad \times \int_{\S} \B_{\Lambda} (v_{N+1}-v_{N+2},\omega)\P_{N+2}\left(t,\bm{V}_{N+2}\right)\d\omega \nonumber \\&
=\frac{\alpha}{2}\int_{\R^{3}} \d v_{\ell+1}\int_{\R^{3}} \d v_{\ell+2}\int_{\S} \B_{\Lambda}(v_{\ell+1}-v_{\ell+2},\omega)f_{\ell+2}\left(t,\bm{V}_{\ell+2}\right)\d\omega.
\end{align}
Putting together \eqref{BBGKYloss}, \eqref{gainALPHAN} and \eqref{gainN+2b} we conclude that the hierarchy solved by the correlation functions $\{f_\ell(t)\}_\ell$ is
exactly \eqref{hierarchyNOTscaledb}.
\end{proof}

\section{BBGKY hierarchy for the rescaled correlation functions}
\label{rescaled}

\setcounter{equation}{0}
\setcounter{theo}{0}
 
\subsection{Weak formulation for the BBGKY hierarchy}

As discussed in the previous sections, we are interested in the limit of the finite particle system, when  the volume $\Lambda$ and the initial number of particles 
$N_0$ go to infinity in such a way that:
\begin{eqnarray}
\label{S5thermod}
\lim_{\Lambda, \,N_0 \to +\infty}\frac{N_0}{\Lambda}=\varrho_0 \in (0, +\infty),
\end{eqnarray}
or, equivalently, using the notation $\varepsilon={\Lambda}^{-1}$ as we did in the Introduction, 
$$N_{0}\,\varepsilon  \longrightarrow\, 1 \qquad \text{ as  $\:N_{0} \to \infty$ and $\varepsilon \to 0^{+}$}.$$
To investigate the limiting behavior of the particle system, we first recall that the rescaled correlation functions $\{f_\ell^{\varepsilon}(t)\}_{\ell=1}^{N_0}$ are given by (see \eqref{rescaledcorr}):
\begin{equation}\label{scalingK}
f_\ell^{\varepsilon}(t,\Vk)
:=\varepsilon^{\ell}\,f_{\ell}(t,\Vk)= \sum_{N=\ell}^\infty
 \frac{\varepsilon^{\ell}}{(N-\ell)!}\int_{\R^3} \d v_{\ell+1}\dots \int_{\R^3}  \d v_{N}\  \P_{N}(\V,t), \quad t \geq 0,
\end{equation}
for any $\Vk=(v_{1},\ldots,v_{\ell}) \in \R^{3\ell}$, $\ell=1,\ldots, N_0.$ 
 
Since the function $f_1$ is the number density function (cf.~Section  \ref{SMarginals}) the rescaled function $f_1^{\varepsilon}$ is the density function associated to the concentration of particles, i.e.  it corresponds to the number of particles per unit volume. 

We proved in Proposition \ref{prp:bbgky} that the correlation functions $\{f_{\ell}(t)\}_{k\geq 1}$ satisfy the BBGKY hierarchy \eqref{hierarchyNOTscaledb}. Analogously, the rescaled correlation functions $\{f_\ell^{\varepsilon}(t)\}_{\ell=1}^{N_0}$ satisfy a rescaled version of the BBGKY hierarchy \eqref{hierarchyNOTscaledb}. It will be convenient to write such a rescaled hierarchy in weak form -- identifying each $f_{\ell}^{\varepsilon}$ with a Radon measure on $\R^{3\ell}$. In order to do this, let us introduce the duality pairing as
$$\big\langle \mu_{\ell},\Phi_{\ell}\big\rangle_{\ell}=\int_{\R^{3\ell}}\Phi_{\ell}(\Vk)\mu_{\ell}(\d\Vk)$$
for any (signed) Radon measure $\mu_{\ell} \in \mathscr{M}(\R^{3\ell})$ and any test-function $\Phi_{\ell} \in \mathcal{C}_{b}(\R^{3\ell})$. When no ambiguity is possible, we simply denote the above pairing as $\big\langle \mu_{\ell},\Phi_{\ell}\big\rangle$ (omitting the last $\ell$-index). We then identify $f_{\ell}^{\varepsilon}(t)$ with a positive Radon measure and write 
$$f_{\ell}^{\varepsilon}(t,\d\Vk):=f_{\ell}^{\varepsilon}(t,\Vk)\d\Vk.$$  A direct consequence of Proposition \ref{prp:bbgky} is the following:

\begin{prp}[\textit{\textbf{Rescaled BBGKY Hierarchy in weak form}}]
For any $\ell \in \{1,\ldots,N_{0}\}$ and any symmetric test function $\Phi_{\ell} \in \mathcal{C}_{b}(\R^{3\ell})$, we have
\begin{multline}\label{BBGKY6}
\big\langle f_{\ell}^{\varepsilon}(t),\Phi_{\ell}\big\rangle_{\ell}=\big\langle f_{\ell}^{\varepsilon}(0),\Phi_{\ell}\big\rangle_{\ell}\\
\phantom{++++} +\varepsilon \sum_{1 \leq i < j \leq \ell}\int_{0}^{t}\d s\int_{\R^{3\ell}}f_{\ell}^{\varepsilon}(s,\d\Vk)\int_{\S} \B(v_i-v_j,\omega) \left[(1-\alpha)\Phi_\ell(\Vk^{i,j})-\Phi_\ell(\Vk)\right]\d\omega\\
+\sum_{i=1}^{\ell}\int_{0}^{t}\d s\int_{\R^{3(\ell+1)}} f_{\ell+1}^{\varepsilon}(s,\d\Vku)
\int_{\S} \B(v_i-v_{\ell+1},\omega) \left[(1-\alpha)\Phi_{\ell}(\widehat{\Vk}^{i,\ell+1})-\Phi_{\ell}(\Vk)\right]\d\omega
\end{multline}
where $\widehat{\Vk}^{i, \ell+1}=(v_1,\ldots,v_{i-1},v_i',v_{i+1},\ldots,v_{\ell})$ with $v_i'=v_i- [ (v_i-v_{\ell+1})\cdot \omega] \omega.$ 
\end{prp}

\begin{proof} We first rewrite \eqref{hierarchyNOTscaledb} for the rescaled correlation functions $f_{\ell}^{\varepsilon}(t)$. This gives the hierarchy satisfied by $f_{\ell}^{\varepsilon}(t)$, i.e. 
 \begin{eqnarray}\label{hierarchyscaled}
&&\partial_t f_\ell^{\varepsilon}(t,\Vk)=\varepsilon \sum_{1 \leq i < j \leq \ell}\int_{\S} \B(v_i-v_j,\omega) \left[(1-\alpha)f_\ell^{\varepsilon}\left(t,\Vk^{i,j}\right)-f_\ell^{\varepsilon}\left(t,\Vk\right)\right]\d\omega\nonumber\\
&&\qquad \ \ +\sum_{i=1}^{\ell}\int_{\R^{3}} \d	v_{\ell+1}\int_{\S} \B(v_i-v_{\ell+1},\omega) \left[(1-\alpha)f_{\ell+1}^{\varepsilon}\left(t,\Vku^{i, \ell+1}\right)-f_{\ell+1}^{\varepsilon}\left(t,\Vku\right)\right]\d\omega.\nonumber\\
&&
\end{eqnarray}
Then, multiplying \eqref{hierarchyscaled} by a test function $\Phi_{\ell}$, integrating by parts and integrating in $t\in [0, T ]$ we obtain \eqref{BBGKY6} in a straightforward way.
\end{proof}

 \begin{rem} Notice that the above velocity vector $\widehat{\Vk}^{i, \ell+1}$ belongs to $\R^{3\ell}$ but, somehow, is deduced from $\Vku$ since it also depends on $v_{\ell+1}.$ 
\end{rem}
 \begin{rem} We will refer to the family of equations (\ref{BBGKY6}) as BBGKY hierarchy 
by  analogy with the system arising  in the framework of classical particle systems. \end{rem}
 By (\ref{inDATA}), (\ref{normalization}) and (\ref{scalingK})  it follows that, at time $t=0$, 
\begin{eqnarray}\label{scaling2}
f_\ell^{\varepsilon}(\Vk,0) =\varepsilon^{\ell}\frac{(N_0)!}{(N_0-\ell)!}  f_0^{\otimes \ell}(\Vk),\ \ \ \ell=1,2,\dots, N_0,
\end{eqnarray}
so that, for every $\ell\ge 1$ we have:
\begin{eqnarray}
\label{scaling6}
\lim_{\substack{\varepsilon \to 0,\,N_0 \to +\infty\\ \varepsilon\,{N_0}\to 1}} ||f_\ell^{\varepsilon}(0)- f_0^{\otimes \ell}||_{1}=0
\end{eqnarray}
where we recall that $\|\cdot\|_{1}$ is the total-variation norm in $\mathscr{M}(\R^{3\ell})$ (corresponding here to the $L^{1}(\R^{3\ell})$-norm). Moreover, assuming that suitable bounds hold for $\{f_{\ell}^{\varepsilon}(t)\}_{\ell}$, it would follow that \eqref{BBGKY6} behaves as
\begin{multline*}
\dfrac{\d}{\d t} \big\langle f_{\ell}^{\varepsilon}(t),\Phi_{\ell}\big\rangle_{\ell}=\mathcal{O}(\varepsilon) 
+ \sum_{i=1}^{\ell}\int_{\R^{3(\ell+1)}} f_{\ell+1}^{\varepsilon}(t,\d\Vku)\\
\int_{\S} \B(v_i-v_{\ell+1},\omega) \left[(1-\alpha)\Phi_{\ell}(\widehat{\Vk}^{i,\ell+1})-\Phi_{\ell}(\Vk)\right]\d\omega.
\end{multline*}
Therefore, we expect that any weak-$\star$ limit $\{g_{\ell}\}_{\ell}$ of $\{f_{\ell}^{\varepsilon}(t)\}_{\ell}$ satisfies
\begin{multline}\label{Eq:HBH}
\dfrac{\d}{\d t} \big\langle g_{\ell} (t),\Phi_{\ell}\big\rangle_{\ell}=\sum_{i=1}^{\ell}\int_{\R^{3(\ell+1)}} g_{\ell+1}(t,\d\Vku)\\
\int_{\S} \B(v_i-v_{\ell+1},\omega) \left[(1-\alpha)\Phi_{\ell}(\widehat{\Vk}^{i,\ell+1})-\Phi_{\ell}(\Vk)\right]\d\omega, \qquad \forall \ell \geq 1.\end{multline}
The above  system of equations will be referred to as the \emph{annihilated Boltzmann hierarchy}.  
It is worth to notice that the difficulty is that, in the limit $N_{0}\to \infty$ and $\varepsilon \to 0$, the above hierarchy is an \emph{infinite} hierarchy (while the rescaled BBGKY hierarchy is actually finite since $\ell \in \{1,\ldots,N_{0}\}$). 
In order to prove this convergence, we first need to establish suitable \emph{a priori} estimates  for the rescaled correlation functions $\{f_{\ell}^{\varepsilon}(t)\}_{\ell}.$

 \subsection{Energy estimates} 
 
 We prove here the following estimates:
 \begin{prp}\label{prp:energy} For any $\varepsilon > 0$ and any $\ell \geq 1$, we define the kinetic energy 
 $$\bm{E}_{\ell}^{\varepsilon}(t)=\int_{\R^{3\ell}}|v_{1}|^{2}f^{\varepsilon}_{\ell}(t,\d\Vk), \qquad t \geq 0$$
 and the mass 
 $$\bm{\varrho}_{\ell}^{\varepsilon}(t)=\int_{\R^{3\ell}}f^{\varepsilon}_{\ell}(t,\d\Vk), \qquad t \geq 0.$$
 Then, 
 $$\bm{E}_{\ell}^{\varepsilon}(t) \leq \bm{E}^{\varepsilon}_{\ell}(0) \qquad \text{ and } \qquad  \bm{\varrho}_{\ell}^{\varepsilon}(t) \leq \bm{\varrho}_{\ell}^{\varepsilon}(0)\qquad \forall t \geq 0.$$\end{prp}
 \begin{proof} In order to prove the kinetic energy estimate, we choose the following test-functions:
 $$\Phi_{\ell}(\Vk)=\E(\Vk)={\ell}^{-1}\sum_{j=1}^{\ell}|v_{j}|^{2},\qquad \Vk=(v_{1},\ldots,v_{\ell}) \in \R^{3\ell}$$
 in the weak formulation of the BBGKY hierarchy \eqref{BBGKY6}. Strictly speaking, $\E$ does not belong to $\mathcal{C}_{b}(\R^{3\ell})$ however, one can  consider the truncated energy 
 $$\E_{r}(\Vk)=\begin{cases} \E(\Vk) \qquad &\text{ if } \E(\Vk) \leq r\\
 r \qquad &\text{ if } \E(\Vk) > r\end{cases}$$
for $r > 0$  -- which belongs to $\mathcal{C}_{b}(\R^{3\ell})$ -- and show that the following estimates are \emph{uniform} with respect to the truncation parameter $r >0$. For simplicity, one proves the result directly for $\E$. We have:
\begin{equation}\begin{split}\label{eq:EqEn}
\dfrac{\d}{\d t}\langle f_{\ell}^{\varepsilon}(t),\E\rangle&=\varepsilon\sum_{1 \leq i < j \leq \ell}\int_{\R^{3\ell}}f_{\ell}^{\varepsilon}(t,\d\Vk)\int_{\S} \B(v_i-v_j,\omega) \left[(1-\alpha)\E(\Vk^{i,j})-\E(\Vk)\right]\d\omega\\
&+\sum_{i=1}^{\ell}\int_{\R^{3(k+1)}} f_{\ell+1}^{\varepsilon}(t,\d\Vku)\int_{\S} \B(v_i-v_{\ell+1},\omega) \left[(1-\alpha)\E(\widehat{\Vk}^{i,\ell+1})-\E(\Vk)\right]\d\omega .
\end{split}
\end{equation}
Notice that $\E(\Vk^{i,j})=\E(\Vk)$ while
\begin{equation*}\begin{split}
\E(\widehat{\Vk}^{i,\ell+1})&={\ell}^{-1}\left(\sum_{j \neq i}|v_{j}|^{2}+|v_{i}'|^{2}\right)\\
&=\E(\Vk)+{\ell}^{-1}\left((v_{i}-v_{\ell+1})\cdot \omega\right)^{2}
-2 {\ell}^{-1}\left((v_{i}-v_{\ell+1})\cdot \omega\right)\left(v_{i}\cdot \omega\right).
\end{split}\end{equation*}
Hence, we can rewrite \eqref{eq:EqEn} as
\begin{equation}\begin{split}\label{eq:EE}
\dfrac{\d}{\d t}\langle f_{\ell}^{\varepsilon}(t),\E\rangle&=-\alpha\,\varepsilon\sum_{1 \leq i < j \leq \ell}\int_{\R^{3\ell}}\E(\Vk)f_{\ell}^{\varepsilon}(t,\d\Vk)\int_{\S} \B(v_i-v_j,\omega) \d\omega\\
&\phantom{++}-\alpha\,\sum_{i=1}^{\ell}\int_{\R^{3(\ell+1)}} \E(\Vk) f_{\ell+1}^{\varepsilon}(t,\d\Vku)\int_{\S} \B(v_i-v_{\ell+1},\omega) \d\omega\\
&\phantom{+}+(1-\alpha){\ell}^{-1}\sum_{i=1}^{\ell}\int_{\R^{3(\ell+1)}} \left[\mathcal{A}^{+}(v_{i},v_{\ell+1})-\mathcal{A}^{-}(v_{i},v_{\ell+1})\right]\,f_{\ell+1}^{\varepsilon}(t,\d\Vku)
\end{split}
\end{equation}
with 
\begin{align*}
&\mathcal{A}^{+}(v,v_{*})=\int_{\S} \B(v-v_{*},\omega)\left((v-v_{*})\cdot \omega\right)^{2}\d\omega,\\&
 \mathcal{A}^{-}(v,v_{*})=2\int_{\S} \B(v-v_{*},\omega)\left((v-v_{*})\cdot \omega\right)\left(v\cdot \omega\right)\d\omega,
 \end{align*}
 for any $(v,v_{*}) \in \R^{6}$.  
Writing $\left[(v-v_{*})\cdot \omega\right]^{2}=\left[(v-v_{*})\cdot \omega\right]\left[v\cdot \omega - v_{\star}\cdot \omega\right]$, it follows  
\begin{multline*}
\int_{\R^{3(\ell+1)}}\mathcal{A}^{+}(v_{i},v_{\ell+1})\,f_{\ell+1}^{\varepsilon}(t,\d\Vku)\\
=\int_{\S}\d\omega\int_{\R^{3(\ell+1)}}\B(v_{i}-v_{\ell+1},\omega)\left[(v_{i}-v_{\ell+1})\cdot \omega\right]\,\left[v_{i} \cdot \omega\right]\,f_{\ell+1}^{\varepsilon}(t,\d\Vku)\\
-\int_{\S}\d\omega\int_{\R^{3(\ell+1)}}\B(v_{i}-v_{\ell+1},\omega)\left[(v_{i}-v_{\ell+1})\cdot \omega\right]\,\left[v_{\ell+1}\cdot\omega\right]f_{\ell+1}^{\varepsilon}(t,\d\Vku)
\end{multline*}
i.e.
\begin{multline*}
\int_{\R^{3(\ell+1)}}\mathcal{A}^{+}(v_{i},v_{\ell+1})\,f_{\ell+1}^{\varepsilon}(t,\d\Vku)=\frac{1}{2}\int_{\R^{3(\ell+1)}}\mathcal{A}^{-}(v_{i},v_{\ell+1})\,f_{\ell+1}^{\varepsilon}(t,\d\Vku)\\
-\int_{\S}\d\omega\int_{\R^{3(\ell+1)}}\B(v_{i}-v_{\ell+1},\omega)\left[(v_{i}-v_{\ell+1})\cdot \omega\right]\left[v_{\ell+1}\cdot \omega\right]\,f_{\ell+1}^{\varepsilon}(t,\d\Vku).\end{multline*}
Now, since $f_{\ell+1}^{\varepsilon}(t)$ is symmetric, we can exchange the role of $v_{i}$ and $v_{\ell+1}$ to get
\begin{multline*}
\int_{\S}\d\omega\int_{\R^{3(\ell+1)}}\B(v_{i}-v_{\ell+1},\omega)\left[(v_{i}-v_{\ell+1})\cdot \omega\right]\left[v_{\ell+1}\cdot \omega\right]\,f_{\ell+1}^{\varepsilon}(t,\d\Vku)\\
=-\int_{\S}\d\omega\int_{\R^{3(\ell+1)}}\B(v_{i}-v_{\ell+1},\omega)\left[(v_{i}-v_{\ell+1})\cdot\omega\right]\,\left[v_{i}\cdot\omega\right]\,f_{\ell+1}^{\varepsilon}(t,\d\Vku)\\
=-\frac{1}{2}\int_{\R^{3(\ell+1)}}\mathcal{A}^{-}(v_{i},v_{\ell+1})f_{\ell+1}^{\varepsilon}(t,\d\Vku)\end{multline*}
i.e.
$$\int_{\R^{3(\ell+1)}}\mathcal{A}^{+}(v_{i},v_{\ell+1})f_{\ell+1}^{\varepsilon}(t,\d\Vku)=\int_{\R^{3(\ell+1)}}\mathcal{A}^{-}(v_{i},v_{\ell+1})f_{\ell+1}^{\varepsilon}(t,\d\Vku).$$
Therefore, \eqref{eq:EE} becomes
\begin{equation*}\begin{split}\label{eq:EEbis}
\dfrac{\d}{\d t}\langle f_{\ell}^{\varepsilon}(t),\E\rangle&=-\alpha\,\varepsilon\sum_{1 \leq i < j \leq \ell}\int_{\R^{3\ell}}\E(\Vk)f_{\ell}^{\varepsilon}(t,\d\Vk)\int_{\S} \B(v_i-v_j,\omega)\d\omega\\
&\phantom{++++}-\alpha\,\sum_{i=1}^{\ell}\int_{\R^{3(\ell+1)}} \E(\Vk) f_{\ell+1}^{\varepsilon}(t,\d\Vku)\int_{\S} \B(v_i-v_{\ell+1},\omega) \d\omega\end{split}\end{equation*}
so that
$$\dfrac{\d}{\d t}\langle f_{\ell}^{\varepsilon}(t),\E\rangle \leq 0$$
which proves the result since $\langle f_{\ell}^{\varepsilon}(t),\E\rangle=\bm{E}_{\ell}^{\varepsilon}(t)$ due to the symmetry of $f_{\ell}^{\varepsilon}(t).$ 
Using the same argument above, it is possible to prove that 
$$\dfrac{\d}{\d t}\bm{\varrho}_{\ell}^{\varepsilon}(t) \leq 0$$
by picking the test-function $\Phi_{\ell}(\Vk)=1$ for all $\Vk$ in \eqref{BBGKY6}.
 \end{proof}
Notice that, from our choice of the initial datum $\P_{N}(0)$ (see \eqref{scaling2}) one has
$$\bm{E}_{\ell}^{\varepsilon}(0)=\varepsilon^{\ell}\frac{(N_0)!}{(N_0-\ell)!} \int_{\R^{3\ell}}|v_{1}|^{2} f_0^{\otimes \ell}(\Vk)\d\Vk=\varepsilon^{\ell}\frac{(N_0)!}{(N_0-\ell)!}E_{0}$$
while
$$\bm{\varrho}_{\ell}^{\varepsilon}(0)=\varepsilon^{\ell}\frac{(N_0)!}{(N_0-\ell)!}.$$
In particular, 
$$\lim_{\substack{\varepsilon \to 0,\,N_0 \to +\infty\\ \varepsilon\,{N_0}\to \varrho_0}}\left(\begin{array}{c}\bm{E}_{\ell}^{\varepsilon}(0) \\\bm{\varrho}_{\ell}^{\varepsilon}(0)\end{array}\right) =\left(\begin{array}{c}   E_{0} \\  1\end{array}\right).$$
This directly yields the following
\begin{cor}\label{cor:energy}  For any $\varepsilon >0$ and $N_{0} \in \N$ even, let $\{f_{\ell}^{\varepsilon}(t)\}_{\ell=1,\ldots,N_{0}}$ be the rescaled correlation functions associated to the unique solution $\{\P_{N}(t)\}_{N}$ to \eqref{PNVN} with initial datum \eqref{inDATA}. Then, 
$$\sup_{t \geq 0}\bm{E}_{\ell}^{\varepsilon}(t)=\sup_{t \geq 0}\int_{\R^{3\ell}}\E(\Vk)f_{\ell}^{\varepsilon}(t,\d\Vku) \leq (N_{0}\varepsilon)^{\ell}E_{0} \qquad \forall \varepsilon >0, N_{0} \in \N, \qquad \ell \geq 1$$
and
$$\sup_{t \geq 0}\bm{\varrho}_{\ell}^{\varepsilon}(t)=\sup_{t \geq 0}\int_{\R^{3\ell}}f_{\ell}^{\varepsilon}(t,\d\Vku) \leq (N_{0}\varepsilon)^{\ell} \qquad \forall \varepsilon >0, N_{0} \in \N, \qquad \ell \geq 1.$$
\end{cor}
 
\subsection{Convergence result} On the basis of the above uniform estimate, we deduce from classical convergence Theorem (see for instance \cite[Theorems 1.40 \& 1.41, p. 65-66]{evans}) the following convergence result:
\begin{prp}\label{prp:conv} For any $\varepsilon >0$ and $N_{0} \in \N$ even, let $\{f_{\ell}^{\varepsilon}(t)\}_{\ell=1,\ldots,N_{0}}$ be the rescaled correlation functions associated to the unique solution $\{\P_{N}(t)\}_{N}$ to \eqref{PNVN} with initial datum \eqref{inDATA}. Then, for any $\ell \geq 1$ and any $t \geq 0$, there exists some positive measure $\bm{\mu}_{\ell}(t)\in \M(\R^{3\ell})$ and a subsequence (still denoted $\{f_{\ell}^{\varepsilon}(t)\}_{\varepsilon >0,N_{0} \in \N}$) such that
\begin{equation}\label{eq:converg}\lim_{\substack{\varepsilon \to 0,\,N_0 \to +\infty\\ \varepsilon\,{N_0}\to 1}}\big \langle f_{\ell}^{\varepsilon}(t), \Phi_{\ell}\big\rangle_{\ell}
=\big \langle \bm{\mu}_{\ell}(t), \Phi_{\ell} \rangle_{\ell} \qquad \forall \Phi_{\ell} \in \mathcal{C}_{0}(\R^{3\ell}).\end{equation}
Moreover, the mapping $t \geq 0 \mapsto \bm{\mu}_{\ell}(t)$ belongs to $\mathcal{C}([0,T),\mathscr{M}(\R^{3\ell}))$ for any $T >0$ and 
\begin{equation*} 
\sup_{t \geq 0}\int_{\R^{3\ell}}\E(\Vk)\bm{\mu}_{\ell}(t,\d\Vk) \leq E_{0}\ , \qquad \sup_{t \geq 0}\int_{\R^{3\ell}}\bm{\mu}_{\ell}(t,\d\Vk) \leq 1.\end{equation*}
\end{prp}
\begin{rem} In other words, up to extracting a subsequence, for any $\ell \geq 1$, the family 
of Radon measure $\{f_{\ell}^{\varepsilon}(t)\}_{\varepsilon >0,N_{0} \in \N}$ converges weakly-$\star$ towards some Radon measure $\bm{\mu}_{\ell}(t)$ as $\varepsilon \to 0,\,N_0 \to +\infty$ with $N_{0}\varepsilon \to 1.$ Moreover, the (weak) limit $\bm{\mu}_{\ell}(t)$ has mass and energy which remain uniformly bounded in time. \end{rem}

\begin{rem}
We notice that, \emph{a priori}, the choice of the converging subsequence depend on $\ell \geq 1$ but, using a diagonal argument, one can construct a \emph{common} subsequence such that the convergence \eqref{eq:converg} hold true \emph{for any} $\ell \geq 1.$
\end{rem}

\begin{proof} The uniform bounds obtained in Corollary \ref{cor:energy} imply in a straightforward way that, for any $\ell \geq 1$ and any compact set $\mathcal{K}\subset \R^{3\ell}$, 
$$\sup_{N_{0}\varepsilon \leq 2}\sup_{t\geq0}f_{\ell}^{\varepsilon}(t,\mathcal{K}) < \infty$$
and we conclude thanks to the well-known weak compactness criterium for measures \cite[Theorem 1.41]{evans}.\end{proof}

\medskip

Actually, we shall need to enlarge the set of test-functions for which the convergence \eqref{eq:converg} holds true. More precisely, we introduce for all $\ell \in \N$ the set
$$\mathcal{W}_{\ell}^{s}=\left\{\Phi_{\ell} \in \C(\R^{3\ell})\;\;\;;\;\left(\E(\Vk)+1\right)^{-s}\Phi_{\ell}\in L^{\infty}(\R^{3\ell})\right\}, \qquad s \in (0,1).$$
Then, the following result holds.
\begin{cor}\label{cor:converg2}
For any $\varepsilon >0$ and $N_{0} \in \N$ even, let $\{f_{\ell}^{\varepsilon}(t)\}_{\ell=1,\ldots,N_{0}}$ be the rescaled correlation functions associated to the unique solution $\{\P_{N}(t)\}_{N}$ to \eqref{PNVN} with initial datum \eqref{inDATA}. Then, for any $\ell \geq 1$ and any $t \geq 0$, 
the positive measure $\bm{\mu}_{\ell}(t)\in \M(\R^{3\ell})$ obtained in Proposition \ref{prp:conv} satisfies
\begin{equation}\label{eq:converg2}\lim_{\substack{\varepsilon \to 0,\,N_0 \to +\infty\\ \varepsilon\,{N_0}\to 1}}\big \langle f_{\ell}^{\varepsilon}(t), \Phi_{\ell}\big\rangle_{\ell}
=\big \langle \bm{\mu}_{\ell}(t), \Phi_{\ell} \rangle_{\ell} \qquad \forall \Phi_{\ell} \in \mathcal{W}_{\ell}^{s}.\end{equation}
Moreover, the convergence is uniform with respect to $t$ in any compact set.
\end{cor}
\begin{proof}
Given $s \in (0,1)$, consider for any $\varepsilon >0$ and $N_{0} \in \N$ the measure
$$\nu_{\ell}^{\varepsilon}(t)=(1+\E(\Vk))^{s}f_{\ell}^{\varepsilon}(t)$$
one deduces from the bounds in Corollary \ref{cor:energy} that
$$\sup_{t \geq 0}\int_{\R^{3\ell}}(1+\E(\Vk))^{1-s}\nu_{\ell}^{\varepsilon}(t,\d\Vk) \leq (N_{0}\varepsilon)^{\ell}(1+E_{0}),$$
i.e.
$$\sup_{t\geq 0}\sup_{\varepsilon{N_0}\leq 2}\int_{\R^{3\ell}}(1+\E(\Vk))^{1-s}\nu_{\ell}^{\varepsilon}(t,\d\Vk) \leq 2^{\ell}(1+E_{0}).$$
Given $R >0$, the set $\mathcal{K}_{R}=\{\Vk \in \R^{3\ell}\;;\,\E(\Vk) \leq R\}$ is a compact subset of $\R^{3\ell}$ and 
$$\sup_{t \geq 0}\sup_{\varepsilon{N_0}\leq 2}\nu_{\ell}^{\varepsilon}(t,\R^{2\ell}\setminus \mathcal{K}_{R}) \leq 2^{\ell}(1+E_{0})(1+R)^{s-1}.$$
Since $s-1 < 0$ and $R >0$ can be chosen arbitrarily large, one sees that the family of measures $\left\{\nu_{\ell}^{\varepsilon}(t)\right\}_{N_{0}\varepsilon \leq 2}$ is tight for any $t \geq 0$. According to Prokhorov's compactness Theorem (see \cite[Theorem 1.7.6, p. 41]{kolok1}), it is relatively compact for the \emph{weak topology} of $\mathscr{M}(\R^{3\ell})$. Therefore, there exists  a subsequence, still denoted $\left\{\nu_{\ell}^{\varepsilon}(t)\right\}_{N_{0}\varepsilon \leq 2}$, and a measure $\overline{\mu}_{\ell}(t)$ such that
\begin{equation*}\label{eq:converg}\lim_{\substack{\varepsilon \to 0,\,N_0 \to +\infty\\ \varepsilon\,{N_0}\to 1}}\big \langle \nu_{\ell}^{\varepsilon}(t), \Psi_{\ell}\big\rangle_{\ell}
=\big \langle \overline{\mu}_{\ell}(t), \Psi_{\ell} \rangle_{\ell} \qquad \forall \Psi_{\ell} \in \mathcal{C}_{b}(\R^{3\ell}).\end{equation*}
Choosing now $\Psi_{\ell}(\Vk)=(1+\E(\Vk))^{-s}\Phi_{\ell}$ with $\Phi_{\ell} \in \mathcal{W}_{\ell}^{s}$ we obtain
\begin{equation*}\label{eq:converg}\lim_{\substack{\varepsilon \to 0,\,N_0 \to +\infty\\ \varepsilon\,{N_0}\to 1}}\big \langle f_{\ell}^{\varepsilon}(t), \Phi_{\ell}\big\rangle_{\ell}
=\big \langle \overline{\mu}_{\ell}(t), \Phi_{\ell} \rangle_{\ell} \qquad \forall \Phi_{\ell} \in \mathcal{W}_{\ell}^{s}.\end{equation*}
Since $\C_{0}(\R^{3\ell}) \subset \mathcal{W}_{\ell}^{s}$, the uniqueness of the weak-$\star$ limit implies that $\overline{\mu}_{\ell}(t)=\bm{\mu}_{\ell}(t)$ and the proof is achieved.
\end{proof}
 
\section{The Annihilated Boltzmann hierarchy}\label{sec:ABH}

We consider 
here the \emph{Annihilated Boltzmann Hierarchy} in weak form. We first introduce the functional space
$$\mathcal{X}  \subset \prod_{k=1}^{\infty}\mathscr{M}_{\mathrm{sym}}(\R^{3k})$$ 
as the subspace of sequences $\bm{\nu}^{\infty}=\left\{{\nu}_{k}\right\}_{k}$ such that $\nu_{k} \in \mathscr{M}_{\mathrm{sym}}(\R^{3k})$ for any $k \geq 1$,   and such that
$$\|\bm{\nu}\|_{\mathcal{X}}:=\sum_{k=1}^{\infty}2^{-k}\|{\nu}_{k}\|_{1,k} <\infty$$
where $\|\cdot\|_{1,k}$ is the weighted total variation norm in $\mathscr{M}(\R^{3k})$ given by
\begin{equation*}\begin{split}
\left\|\nu_{k}\right\|_{1,k}&=\sup\left\{\bigg|\langle \nu_{k},(1+|\cdot|) \Phi_{k}\rangle_{k}\bigg|\,;\,\Phi_{k} \in \C_{b}(\R^{3k}))\;;\,\|\Phi_{k}\|_{\infty} \leq 1\right\}\\
&=\underset{\|\Phi_{k}\|_{\infty} \leq 1}{\sup_{\Phi_{k} \in \C_{b}(\R^{3k})}}\int_{\R^{3k}}\left(1+|\bm{V}_{k}|\right) \,|\Phi_{k}(\bm{V}_{k})|\nu_{k}(\d\bm{V}_{k}).\end{split}\end{equation*}
Recall that we consider here only \emph{symmetric} measures, i.e.
$$\int_{\R^{3k}}\Phi_{k}(\bm{V}_{k})\nu_{k}(\d \bm{V}_{k})=\int_{\R^{3k}}\Phi_{k}(\bm{V}_{\sigma(k)})\nu_{k}(\d \bm{V}_{k})$$
for any permutation $\sigma$ of $\{1,\ldots,k\}.$

We notice that $(\mathcal{X} ,\|\cdot\|_{\mathcal{X}})$ is a Banach space. Notice that, if $\bm{\nu}=\left\{{\nu}_{k}\right\}_{k} $ is such that $\sup_{k} \|\nu_{k}\|_{1,k} < \infty$ then $\bm{\nu} \in \mathcal{X}.$  

We now define the following notion of solutions to the \emph{Annihilated Boltzmann Hierarchy} that we will denote for shortness ABH in the sequel.
\begin{defi}[\textit{\textbf{Weak solution to the ABH}}] \label{def:weaksol}
Assume that $\B$ satisfy Assumption \ref{hyp1} with $\gamma \in [0,1].$ Given $T >0$,  we say that a family  
$$\bm{\nu}^{\infty}=\left\{{\nu}_{k}\right\}_{k} \in L^{\infty}([0,T)\,;\mathcal{X})$$
is a (weak) solution to the ABH 
if for any $\Phi_{k} \in \C_{0}(\R^{3k})$, $k \geq 1$,  the following identity holds:
\begin{multline}\label{sol:BHweak}
  \langle \nu_k(t), \Phi_{k}\rangle_{k}=\langle \nu_{k}(0),\Phi_{k}\rangle_{k}+
  \sum_{i=1}^{k} \int_{0}^{t}\d s \int_{\R^{3(k+1)}} \nu_{k+1}(s,\d\bm{V}_{k+1})\times\\
  \times\int_{\S} \B(v_i-v_{k+1},\omega) \left[(1-\alpha)\Phi_{k}(\widehat{\bm{V}_{k}}^{i,k+1})-\Phi_{k}(\bm{V}_{k})\right]\d\omega \qquad \forall t \in [0,T).
\end{multline}
\end{defi}

\begin{rem}\label{rem:Gk} For any $k \geq 1,$ $\alpha \in [0,1)$ we  introduce the operator
$$\Gk^{\alpha}\::\:\C_{b}(\R^{3k}) \to \C(\R^{3(k+1)})$$ such that 
$$\Gk^{\alpha} \Phi_{k}(\bm{V}_{k+1})=\sum_{i=1}^{k}\int_{\S}\B(v_i-v_{k+1},\omega) \left[(1-\alpha)\Phi_{k}(\widehat{\bm{V}_{k}}^{i,k+1})-\Phi_{k}(\bm{V}_{k})\right]\d\omega$$
for any $\Phi_{k} \in \C_{b}(\R^{3k}).$ Notice that $\Gk^{\alpha} \Phi_{k}$ is continuous over $\R^{3k+3}$ but no longer bounded whenever $\gamma >0$ since $\Sigma_{\B}$ is unbounded. One can reformulate \eqref{sol:BHweak} as
\begin{multline}\label{sol:BHweak2}\langle \nu_k(t), \Phi_{k}\rangle_{k}=\langle \nu_{k}(0),\Phi_{k}\rangle_{k}+
  \int_{0}^{t} \langle \nu_{k+1}(s), \Gk^{\alpha}\Phi_{k} \rangle_{k+1}\d s \\
  \qquad \forall t \in [0,T),\quad k \geq 1, \qquad \Phi_{k} \in \C_{0}(\R^{3k}).
\end{multline}
\end{rem}
The existence of a solution to the annihilated Boltzmann hierarchy is deduced from the Cauchy theory of Eq. \eqref{BE}. More precisely, we have the following

\begin{prp}\label{theo:existence}
Let $f_{0}$ be a non-negative probability distribution satisfying \eqref{normalization}. Let be $\alpha \in (0,1)$ and let $\B$ satisfy Assumption \ref{hyp1} with $\gamma \in [0,1]$.  
Then, $\bm{\nu}^{\infty}=\{\nu_{k}\}_{k} \in L^{\infty}([0,\infty),\mathcal{X})$ with 
$$\nu_{k}(t)= f(t)^{\otimes k}, \qquad t \geq 0, \qquad k \geq 1$$ 
is a solution to the ABH in the sense of Definition \ref{def:weaksol} with initial datum $\nu_{k}(0)=f_{0}^{\otimes k}$ for all $k \geq 1$ where $f(t)$ is the unique solution to \eqref{BE} with initial datum $f(0)=f_{0}$. 
\end{prp}

\begin{proof} 
 The existence and uniqueness of a solution $f(t)$ to \eqref{BE} with initial datum $f(0)=f_{0}$ is granted from \cite[Theorem 1.9]{BaLo}. Moreover, the solution $f$ satisfies $f \in \C([0,\infty)\,;\,L^{1}_{2}(\R^{3})) \cap L^{1}_{\mathrm{loc}}((0,\infty);L^{1}_{3}(\R^{3}))$
where 
$$L^{1}_{p}(\R^{3})=\{g\in L^{1}(\R^{3})\;;\; \|g\|_{L^{1}_{p}}:=\int_{\R^{3}}|g(v)| \left(1+|v|^{2}\right)^{p/2}\d v < \infty\},\qquad p \geq 0$$
and, additionally,
$$\int_{\R^{3}}f(t,v)\d v \leq 1, \qquad \int_{\R^{3}}f(t,v)\d v \leq E_{0}, \qquad \forall t \geq 0.$$
Introducing, for any $t \geq 0$ and any $k \geq 1$, $\nu_{k}(t)=f^{\otimes k}(t)$, a direct inspection shows that $\bm{\nu}^{\infty}=\{\nu_{k}\}_{k} \in \mathcal{X}$ is a weak solution to \eqref{sol:BHweak}. 
We now recall this argument following the strategy proposed in  
\cite{MM}. We introduce $\mathbb{B}_{\alpha}(f,g)=(1-\alpha)\Q(f,g)-\Q_{-}(f,f)$ the nonlinear annihilated Boltzmann operator. As in \cite{MM}, we introduce also
$\mathbf{B}^{*}_{\alpha}\::\:\psi \in \C_{b}(\R^{3}) \mapsto \mathbf{B}^{*}(\psi) \in \C(\R^{6})$ through
$$\mathbf{B}^{*}_{\alpha}(\psi)(v,v_{*})=\int_{\S}\B(v-v_{*},\omega)\left[(1-\alpha)\psi(v')-\psi(v)\right]\d \omega$$
where $v'$ is the post-collisional velocity associated to the triple $(v,v_{*},\omega)$. Notice that, if $f \in L^{1}(\R^{3})$ is given, the following identity holds:
$$\langle f \otimes f,\mathbf{B}^{*}_{\alpha}(\psi)\rangle_{2}=\int_{\R^{3}}f(v)\d v \int_{\R^{3}} f(v_{*})\mathbf{B}^{*}_{\alpha}(\psi)(v,v_{*})\d v_{*}
=\int_{\R^{3}}\psi(v)\mathbb{B}_{\alpha}(f,f)\d v=\langle \mathbb{B}_{\alpha}(f,f),\psi\rangle_{1}.$$
Then, for $\Phi_{k}(\bm{V}_{k})=\varphi_{1}(v_{1})\ldots\varphi_{k}(v_{k})$, one has 
$$\Gk^{\alpha} \Phi_{k}(\bm{V}_{k+1})=\sum_{i=1}^{k}\left(\prod_{j\neq i}\varphi_{j}(v_{j})\right)\mathbf{B}^{*}_{\alpha}(\varphi_{i})(v_{i},v_{k+1})$$
and, for a given $f \in L^{1}(\R^{3})$,

\begin{equation*}\begin{split}
\langle f^{\otimes (k+1)},\Gk^{\alpha} \Phi_{k}\rangle_{k+1}&=\sum_{i=1}^{k}\int_{\R^{3(k-1)}}\left(\prod_{j\neq i}\varphi_{j}(v_{j})f(v_{j})\right)\d v_{1}\ldots \d v_{i-1}\d v_{i+1}\ldots \d v_{k}\\
&\phantom{+++++}\int_{\R^{3}} \d v_{i} \int_{\R^{3}} \d v_{k+1} f(v_{i})f(v_{k+1})\mathbf{B}^{*}_{\alpha}(\varphi_{i})(v_{i},v_{k+1}) \d v_{k+1}\\
&=\sum_{i=1}^{k}\int_{\R^{3(k-1)}}\left(\prod_{j\neq i}\varphi_{j}(v_{j})f(v_{j})\right)\varphi_{i}(v_{i})\mathbb{B}_{\alpha}(f,f)(v_{i})\d\bm{V}_{k}\\
&=\sum_{i=1}^{k}\left(\prod_{j\neq i}\int_{\R^{3}}\varphi_{j}(v_{j})f(v_{j})\d v_{j}\right)\int_{\R^{3}}\varphi_{i}(v_{i})\mathbb{B}_{\alpha}(f,f)(v_{i})\d v_{i}\\
&=\sum_{i=1}^{k}\left(\prod_{j\neq i} \langle f,\varphi_{j}\rangle_{1}\right)\langle \mathbb{B}_{\alpha}(f,f),\varphi_{i}\rangle_{1}
\end{split}\end{equation*}

In particular, if $f(t)=f(t,v)$ is the unique solution to the annihilated Boltzmann equation \eqref{BE}, then $f(t)^{\otimes k}$ satisfies, for all
\begin{equation}\label{eq:tensor}
\Phi_{k}(\bm{V}_{k})=\varphi_{1}(v_{1})\ldots\varphi_{k}(v_{k})\qquad \varphi_{i} \in \C_{b}(\R^{3}), \qquad i=1,\ldots,k.\end{equation}
the following
\begin{equation*}\begin{split}
\dfrac{\d}{\d t}\langle f(t)^{\otimes k},\Phi_{k}\rangle_{k}&=\sum_{i=1}^{k}\left(\prod_{j\neq i} \int_{\R^{3}}f(t,v_{j})\varphi_{j}(v_{j})\d v_{j}\right)\int_{\R^{3}}\partial_{t}f(t,v_{i})\varphi_{i}(v_{i})\d v_{i}\\
&=\sum_{i=1}^{k}\left(\prod_{j\neq i} \langle f(t),\varphi_{j}\rangle_{1}\right)\langle \mathbb{B}_{\alpha}(f(t),f(t)),\varphi_{i}\rangle_{1}\\
&=\langle f(t)^{\otimes (k+1)},\Gk^{\alpha} \Phi_{k}\rangle_{k+1}\end{split}\end{equation*}
which shows that $f(t)^{\otimes k}$ satisifes \eqref{sol:BHweak2} for \textit{\textbf{tensorized}} test-function $\Phi_{k}$ of the form \eqref{eq:tensor}. 
We observe that Stone-Weierstrass theorem guarantees that the class of tensorized test-functions is dense in $\C_{b}(\R^{3k})$. Therefore, we deduce that $f(t)^{\otimes k}$ is a weak solution to ABH.
\end{proof}

\subsection{Solutions to the ABH as limit point of the rescaled correlation functions}

We show here that any weak limit constructed in Proposition \ref{prp:conv} is actually a solution to the Annihilated Boltzmann Hierarchy:
\begin{theo}\label{main}
 For any $\varepsilon >0$ and $N_{0} \in \N$ even, let $\{f_{\ell}^{\varepsilon}(t)\}_{\ell=1,\ldots,N_{0}}$ be the rescaled correlation functions associated to the unique solution $\{\P_{N}(t)\}_{N}$ to \eqref{PNVN} with initial datum \eqref{inDATA}. Then, any limit point of $\{f_{\ell}^{\varepsilon}(t)\}_{\ell=1,\ldots,N_{0}}$ in the sense of \eqref{eq:converg} is a weak solution to ABH. More precisely, if given $\ell \geq 1$ and $t \geq 1$, $\bm{\mu}_{\ell}(t) \in \M(\R^{3\ell})$ is the limit of  a subsequence (still denoted $\{f_{\ell}^{\varepsilon}(t)\}_{\varepsilon >0,N_{0} \in \N}$) in the sense that
\begin{equation}\label{convell}\lim_{\substack{\varepsilon \to 0,\,N_0 \to +\infty\\ \varepsilon\,{N_0}\to 1}}\big \langle f_{\ell}^{\varepsilon}(t), \Phi_{\ell}\big\rangle_{\ell}
=\big \langle \bm{\mu}_{\ell}(t), \Phi_{\ell} \rangle_{\ell} \qquad \forall \ell \geq 1,\quad \Phi_{\ell} \in \mathcal{C}_{0}(\R^{3\ell})\end{equation}then, $\bm{\mu}^{\infty}=\left\{\bm{\mu}_{\ell}(\cdot)\right\}_{\ell} \in L^{\infty}([0,\infty),\mathcal{X})$ is a solution to ABH in the sense of Definition \ref{def:weaksol}.\end{theo}

\begin{proof} The strategy -- as explained in the introduction -- is inspired by \cite[Section 8]{MM}. Let $\alpha \in (0,1)$ and $T >0 $ be given. Define
$$\bm{\mu}^{\infty}(t)=\left\{\bm{\mu}_{\ell}(t)\right\}_{\ell} \in \prod_{\ell=1}^{\infty}\mathscr{M}^{+}(\R^{3\ell}), \qquad t \in [0,T)$$
where, for any $\ell \geq 1$ and any $t \in [0,T)$, $\bm{\mu}_{\ell}(t)$ satisfies \eqref{eq:converg}. We want to prove that $\bm{\mu}^{\infty}$ is a weak solution to \eqref{sol:BHweak}.

\noindent \textit{First Step:} The fact that the mapping $t \in [0,T) \mapsto \bm{\mu}^{\infty}(t)$ belongs to $L^{\infty}([0,T),\mathcal{X})$ is deduced directly from the estimate
$$\sup_{t \geq 0}\int_{\R^{3\ell}}\bm{\mu}_{\ell}(t,\d\Vk)=\sup_{t\geq 0}\|\bm{\mu}_{\ell}(t)\|_{1,\ell} \leq 1$$
obtained in Proposition \ref{prp:conv}. Notice that the second part of Proposition \ref{prp:conv} actually asserts that the mapping $t \in [0,T) \mapsto \bm{\mu}^{\infty}(t) \in \mathcal{X}$ is continuous. 

\noindent \textit{Second Step:} Our goal is to  show that $\bm{\mu}^{\infty}(t)$ satisfies \eqref{sol:BHweak}. We first notice that, according to \eqref{scaling6}, we have
$$\bm{\mu}_{\ell}(0)=f_{0}^{\otimes \ell} \qquad \forall \ell \geq 1.$$
We now fix $t \in [0,T)$ and  $\Phi_{\ell} \in \C_{0}(\R^{3\ell})$. Recall from \eqref{BBGKY6} that
\begin{multline}\label{fellvare} 
\big\langle f_{\ell}^{\varepsilon}(t),\Phi_{\ell}\big\rangle_{\ell}=\big\langle f_{\ell}^{\varepsilon}(0),\Phi_{\ell}\big\rangle_{\ell} 
+ \int_{0}^{t} \langle f_{\ell+1}^{\varepsilon}(s), \bm{\Gamma}_{\ell,\ell+1}^{\alpha}\Phi_{\ell} \rangle_{\ell+1}\d s\\
\phantom{++++} +\varepsilon \sum_{1 \leq i < j \leq \ell}\int_{0}^{t}\d s\int_{\R^{3\ell}}f_{\ell}^{\varepsilon}(s,\d\Vk)\int_{\S} \B(v_i-v_j,\omega) \left[(1-\alpha)\Phi_\ell(\Vk^{i,j})-\Phi_\ell(\Vk)\right]\d\omega.\end{multline} 
Clearly, the term on the left-hand side and the first term on the right-hand side converge respectively to $\big \langle \bm{\mu}_{\ell}(t),\Phi_{\ell}\big\rangle_{\ell}$ and $\big \langle \bm{\mu}_{\ell}(0),\Phi_{\ell}\big\rangle_{\ell}$. Let us consider the second-term on the right-hand side, i.e.
$$\mathcal{G}_{\ell}^{\varepsilon}:=\int_{0}^{t} \langle f_{\ell+1}^{\varepsilon}(s), \bm{\Gamma}_{\ell,\ell+1}^{\alpha}\Phi_{\ell} \rangle_{\ell+1}\d s.$$
One sees easily that
\begin{multline*}
|\bm{\Gamma}_{\ell,\ell+1}^{\alpha}\Phi_{\ell}(\Vku)| \leq 2\|\Phi_{\ell}\|_{\infty}\,\sum_{i=1}^{\ell}\Sigma_{\B}(v_{i}-v_{\ell+1})
\leq 2\|\Phi_{\ell}\|_{\infty}C_{\B}\sum_{i=1}^{\ell}|v_{i}-v_{\ell+1}|^{\gamma} \\
\leq 2C_{\B}\ell(\ell+1)\|\Phi_{\ell}\|_{\infty}\E(\Vku)^{\frac{\gamma}{2}}\qquad \forall \Vku \in \R^{3(\ell+1)}
\end{multline*}
where we used Holder's inequality for the last estimate. This shows that $\bm{\Gamma}_{\ell,\ell+1}^{\alpha}\Phi_{\ell} \in \mathcal{W}_{\ell+1}^{\frac{\gamma}{2}}$ and therefore, using Corollary \ref{cor:converg2}, we get that
$$\lim_{\substack{\varepsilon \to 0,\,N_0 \to +\infty\\ \varepsilon\,{N_0}\to 1}}\mathcal{G}_{\ell}^{\varepsilon}=\int_{0}^{t}\langle \bm{\mu}_{\ell+1}(s),\bm{\Gamma}_{\ell,\ell+1}^{\alpha}\Phi_{\ell} \rangle_{\ell+1}\d s.$$
 We now investigate the third term in \eqref{fellvare}. We introduce, using the notations of Section \ref{sec:reform},
\begin{multline*}
\mathcal{L}^{\star}_{\ell}\Phi_{\ell}(\Vk)=\sum_{1 \leq i < j \leq \ell}\int_{\S} \B(v_i-v_j,\omega) \left[(1-\alpha)\Phi_\ell(\Vk^{i,j})-\Phi_\ell(\Vk)\right]\d\omega\\
\qquad \forall \Vk \in \R^{3\ell}, \quad \Phi_{\ell} \in \C_{0}(\R^{3\ell}),\end{multline*}
so that the third term in the right-hand side of  \eqref{fellvare} is 
$$\mathcal{Z}_{\ell}^{\varepsilon}:=\varepsilon  \int_{0}^{t} \langle f_{\ell}^{\varepsilon}(s), \mathcal{L}^{\star}_{\ell}\Phi_{\ell}\rangle_{\ell}\d s.$$
As in the previous step, it is easy to show that
$$
|\mathcal{L}^{\star}_{\ell}\Phi_{\ell}(\Vk)| \leq 2\|\Phi_{\ell}\|_{\infty}\,\sum_{1 \leq i < j \leq \ell}\Sigma_{\B}(v_{i}-v_{j})
\leq 2C_{\B}\|\Phi_{\ell}\|_{\infty}\sum_{1 \leq i < j \leq \ell}\left(|v_{i}|^{\gamma}+|v_{j}|^{\gamma}\right).$$
Thus, using Young's inequality, there is $C(\gamma,\B) >0$ such that 
$$|\mathcal{L}^{\star}_{\ell}\Phi_{\ell}(\Vk)| \leq C(\gamma,\B)\ell^{2}\left[\E(\Vk)  + 1\right]\|\Phi_{\ell}\|_\infty, \qquad \Vk \in \R^{3\ell}.$$
Relying on the bounds provided in Corollary \ref{cor:energy}, we get
$$\left|\mathcal{Z}^{\varepsilon}_{\ell}\right| \leq \varepsilon\,t C(\gamma,\B)\ell^{2}\|\Phi_{\ell}\|_{\infty}\left(E_{0}+1\right)(N_{0}\varepsilon)^{\ell}$$
which results in $\lim_{\substack{\varepsilon \to 0,\,N_0 \to +\infty\\ \varepsilon\,{N_0}\to 1}}\mathcal{Z}_{\ell}^{\varepsilon}=0.$  We finally obtain that 
$$
\big \langle \bm{\mu}_{\ell}(t),\Phi_{\ell}\big\rangle_{\ell} = \big \langle \bm{\mu}_{\ell}(0),\Phi_{\ell}\big\rangle_{\ell}+\int_{0}^{t}\langle \bm{\mu}_{\ell+1}(s),\bm{\Gamma}_{\ell,\ell+1}^{\alpha}\Phi_{\ell} \rangle_{\ell+1}\d s$$
i.e. $\{\bm{\mu}_{\ell}(t)\}_{\ell}$ is a solution to the annihilated Boltzmann hierarchy. This concludes the proof.\end{proof}

\subsection{Uniqueness of a solution for $\Sigma_{\B}$ bounded}

We consider in this section the case in which
$$\Sigma_{\B}(v-v_{*})=\int_{\S}\B(v-v_{*},\omega)\d\omega$$ 
is bounded. Then, we can prove that there is a unique solution to the ABH. This implies that  the propagation of chaos for \eqref{BE} holds, yielding the proof of Theorem \ref{main1}. Namely, we prove the following
\begin{theo}\label{unique}
Assume that $\Sigma_{\B}$ is bounded, i.e.~there exists $C_{\B} >0$ such that
\begin{equation}\label{eq:boudedkernel}
\int_{\S} \B(v-v_{*},\omega)\d \omega \leq C_{\B} \qquad \forall v,v_{*} \in \R^{3}.
\end{equation}
Then, for any $\varrho_{0} >0$ and any $\bm{\nu}^{\infty}(0)=\left(\nu_{\ell}(0)\right)_{\ell} \in \mathcal{X}$ satisfying $\nu_{\ell}(0) \in \mathscr{M}^{+}(\R^{3\ell})$ and
\begin{equation}\label{hypnu0}\|\nu_{\ell}(0)\|_{1,\ell} \leq \varrho_{0}^{\ell}, \qquad \forall \ell \geq 1,\end{equation}
there is \emph{at most one} weak solution to the Annihilated Boltzmann Hierarchy \eqref{sol:BHweak} in the sense of Definition \ref{def:weaksol}.\end{theo}
The proof will be based on the following Lemma
\begin{lem}\label{lem:itera}
Assume that $\Sigma_{\B}$ is bounded. Then, for any $k \geq 1$ and any $\alpha \in [0,1)$ the operator $\Gk^{\alpha}\::\:\C_{b}(\R^{3k}) \to \C(\R^{3(k+1)})$ defined in Remark \ref{rem:Gk} by
$$\Gk^{\alpha} \Phi_{k}(\bm{V}_{k+1})=\sum_{i=1}^{k}\int_{\S}\B(v_i-v_{k+1},\omega) \left[(1-\alpha)\Phi_{k}(\widehat{\bm{V}_{k}}^{i,k+1})-\Phi_{k}(\bm{V}_{k})\right]\d\omega, \quad \forall \Phi_{k} \in \C_{b}(\R^{3k})$$
has its range in $\C_{b}(\R^{3k+1})$ and is bounded, i.e. 
$$\Gk^{\alpha} \in \mathscr{B}(\C_{b}(\R^{3k}),\C_{b}(\R^{3(k+1)}))$$
with
\begin{equation}\label{eq:stimaGk}
\|\Gk^{\alpha}\|_{\mathrm{op}}:=\|\Gk^{\alpha}\|_{\mathscr{B}(\C_{b}(\R^{3k}),\C_{b}(\R^{3(k+1)}))} \leq (2-\alpha)k\|\Sigma_{\B}\|_{\infty} \qquad \forall k \geq 1.\end{equation}
\end{lem}
\begin{proof} Let $k \geq 1$ be given and $\Phi_{k} \in \C_{b}(\R^{3k}))$ be fixed. One has
\begin{equation*}\begin{split}
\left|\Gk^{\alpha}\Phi_{k}(\bm{V}_{k+1})\right|
&\leq \sum_{i=1}^{k}\int_{\S}\B(v_{i}-v_{k+1},\omega)\left|(1-\alpha)\Phi_{k}(\widehat{\bm{V}_{k}}^{i,k+1})-\Phi_{k}(\bm{V}_{k})\right|\d\omega\\
&\leq (2-\alpha)\|\Phi_{k}\|_{\infty}\sum_{i=1}^{k}\int_{\S}\B(v_{i}-v_{k+1},\omega)\d \omega, \qquad \forall \bm{V}_{k+1} \in \R^{3(k+1)}\end{split}\end{equation*}
where $\|\Phi_{k}\|_{\infty}=\sup_{\bm{V}_{k} \in \R^{3k}}\left|\Phi_{k}(\bm{V}_{k})\right|.$ Therefore, under the assumption that $\Sigma_{\B}$ is bounded by some positive constant $C_{\B}$ we deduce that
$$\|\Gk^{\alpha}\Phi_{k}\|_{\infty}=\sup_{\bm{V}_{k+1} \in \R^{3k+3}}\left|\Gk^{\alpha}\Phi_{k}(\bm{V}_{k+1})\right|\leq (2-\alpha)k\,C_{\B}\|\Phi_{k}\|_{\infty}$$
which is the desired result.\end{proof}

\begin{proof}[Proof of Theorem \ref{unique}] We adapt here the proof of \cite{fede-esco} given for coagulation equation. Let $\bm{\nu}^{\infty}(0)=\left(\bm{\nu}_{k}(0)\right)_{k} \in \mathcal{X}$ be fixed and let $\bm{\nu}^{(1)},\bm{\nu}^{(2)}$ be two solutions to \eqref{sol:BHweak} associated to the initial datum $\bm{\nu}^{\infty}(0).$ Using the formulation of the ABH given in \eqref{sol:BHweak2} with $\Phi_{k}=\mathbf{1}$ one sees that, for any $T >0$,
$$\langle \nu_{k}^{(i)}(t),\mathbf{1}\rangle_{k}=\langle \nu_{k}(0),\mathbf{1}\rangle_{k} + \int_{0}^{t}\langle \nu_{k+1}^{(i)}(s),\Gk^{\alpha}\mathbf{1}\rangle_{k+1}\d s, \qquad t \in [0,T)$$
with $\Gk^{\alpha}\mathbf{1}(\bm{V}_{k+1})=-\alpha\sum_{i=1}^{k}\B(v_{i}-v_{k+1},\omega)\d\omega \leq 0$ for any $\bm{V}_{k+1} \in \R^{3k+1}$. Thus, for any $t \geq 0$,
$$\langle \nu_{k}^{(i)}(t),\mathbf{1}\rangle_{k} \leq \langle \nu_{k}(0),\mathbf{1}\rangle_{k}, \qquad \forall t \geq 0, i=1,2.$$
Using that the measures $\nu_{k}^{(1)}(t)$ and $\nu_{k}^{(2)}(t)$ are nonnegative,  the following estimate for the total variation norm $\|\cdot\|_{1,k}$ follows:
\begin{equation}\label{eq:TVnu}
\|\nu_{k}^{(i)}(t)\|_{1,k} \leq \|\nu_{k}(0)\|_{1,k} \qquad \forall t \geq 0, \qquad \forall k \geq 1.\end{equation}
On the other hand, using again the formulation \eqref{sol:BHweak2}, it holds
$$\langle \nu_k^{(1)}(t)-\nu_{k}^{(2)}(t), \Phi_{k}\rangle_{k}=  \int_{0}^{t} \langle \nu_{k+1}^{(1)}(s)-\nu_{k+1}^{(2)}(s), \Gk^{\alpha}\Phi_{k} \rangle_{k+1}\d s,$$
  for any $t \in [0,T),$ $k \geq 1$ and $\Phi_{k} \in \C_{0}(\R^{3k}).$ We set $\bm{\beta}_{k}(t)=\nu_{k}^{(1)}(t)-\nu_{k}^{(2)}(t)$ for any $t \in [0,T)$ and any $k \geq 1$, and consider  a sequence $\left\{\Phi_{k}\right\}_{k \geq 1}$ with $\Phi_{k} \in \C_{0}(\R^{3k})$ for any $k \geq 1.$   One has
\begin{equation}\label{eq:betak}
\langle \bm{\beta}_{k}(t),\Phi_{k}\rangle_{k}=\int_{0}^{t}\langle \bm{\beta}_{k+1}(s)\,,\,\Gk^{\alpha}\Phi_{k}\rangle_{k+1}\d s, \qquad t \in [0,T), \qquad \forall k \geq 1.\end{equation}
According to Lemma \ref{lem:itera}, $\Gk^{\alpha} \in \mathscr{B}(\C_{b}(\R^{3k}),\C_{b}(\R^{3k+1}))$ which allows to iterate the above formula to get
$$\langle \bm{\beta}_{k}(t),\Phi_{k}\rangle_{k}=\int_{0}^{t}\d s \int_{0}^{s}\langle \beta_{k+2}(s_{1})\,,\,\bm{\Gamma}_{k+1;k+2}^{\alpha}\,\Gk^{\alpha}\Phi_{k}\rangle_{k+2}\d s_{1}$$
and, iterating again, for any $n \geq 1$, we get
\begin{multline*}
\langle \bm{\beta}_{k}(t),\Phi_{k}\rangle_{k}=\\
\int_{0}^{t}\d s \int_{0}^{s}\d s_{1}\ldots\int_{0}^{s_{n}}
\langle \beta_{k+n+1}(s_{n})\,,\,
\bm{\Gamma}_{k+n;k+n+1}^{\alpha}\bm{\Gamma}_{k+n-1;k+n}^{\alpha}\ldots \Gk^{\alpha}\Phi_{k}\rangle_{k+n+1}\d s_{n+1}.\end{multline*}
Therefore, using now \eqref{eq:stimaGk}, we have
\begin{multline*}
\left\|\bm{\Gamma}_{k+n;k+n+1}^{\alpha}\bm{\Gamma}_{k+n-1;k+n}^{\alpha}\ldots \Gk^{\alpha}\Phi_{k}\right\|_{\C_{b}(\R^{3(k+n+1)})} 
\\
\leq \left((2-\alpha)\|\Sigma_{\B}\|_{\infty}\right)^{n+1}\,\frac{(k+n)!}{(k-1)!}\|\Phi_{k}\|_{\C_{b}(\R^{3k})}
\end{multline*}
from which we deduce that
$$\left|\langle \bm{\beta}_{k}(t),\Phi_{k}\rangle_{k}\right| \leq \frac{t^{n+2}}{(n+2)!}\left((2-\alpha)\|\Sigma_{\B}\|_{\infty}\right)^{n+1}\,\frac{(k+n)!}{(k-1)!}\|\Phi_{k}\|_{\C_{b}(\R^{3k})}\,\sup_{s\in [0,t]}\left\|\bm{\beta}_{k+n+1}(s)\right\|_{1,k+n+1}$$
where $\|\cdot\|_{1,k+n+1}$ is the total variation norm on $\mathscr{M}(\R^{3(k+n+1)})$. Thus, the sequence $\{\bm{\beta}_{k}\}_{k}$ satisfies
\begin{equation}\label{Eq:betak-n}
\left\|\bm{\beta}_{k}(t)\right\|_{1,k} \leq \frac{t^{n+2}}{(n+2)!}\left((2-\alpha)\|\Sigma_{\B}\|_{\infty}\right)^{n+1}\,\frac{(k+n)!}{(k-1)!}\,\sup_{s\in [0,T)}\left\|\bm{\beta}_{k+n+1}(s)\right\|_{1,k+n+1},\end{equation}
for any $t \in [0,T),$ and for any $k,n \geq 1.$ Using now \eqref{eq:TVnu}, we have
\begin{equation*}\begin{split}
\sup_{s \in [0,T)}\left\|\bm{\beta}_{k+n+1}(s)\right\|_{1,k+n+1} &\leq \sup_{s \in [0,T)}\left(\|\nu_{k+n+1}^{(1)}(s)\|_{1,k+n+1} + \|\nu_{k+n+1}^{(2)}(s)\|_{1,k+n+1}\right)\\
 &\leq 2\|\nu_{k+n+1}(0)\|_{1,k+n+1} \leq 2\varrho_{0}^{n+1+k}\end{split}\end{equation*}
where we used  Assumption \eqref{hypnu0} for the last estimate. Inserting this in \eqref{Eq:betak-n} this yields finally
$$\left\|\bm{\beta}_{k}(t)\right\|_{1,k} \leq 2\frac{t^{n+2}}{(n+2)!}\left((2-\alpha)\|\Sigma_{\B}\|_{\infty}\right)^{n+1}\,\frac{(k+n)!}{(k-1)!}\,\varrho_{0}^{k+n+1}, \qquad \forall k \geq 1, \forall n \geq 1, \qquad t \in [0,T).$$
Notice that (see \cite[Eq. (3.11)]{fede-esco})
$$\frac{(k+n)!}{(n+2)!(k-1)!} \leq \frac{2^{k+n}}{n+2} \leq 2^{k+n-1}, \qquad \forall k,n \geq 1.$$
Hence, picking $T>0$ small enough so that 
$$a:=2\varrho_{0}T\,(2-\alpha)\|\Sigma_{\B}\|_{\infty} < 1,$$
we get
$$\sup_{t \in [0,T)}\|\bm{\beta}_{k}(t)\|_{1,k} \leq 2^{k-1}T\,\varrho_{0}^{k}\,a^{n+1} \qquad k \geq 1, n \geq 1$$
and, letting $n \to \infty$ this shows that
$$\sup_{t \in [0,T)}\|\bm{\beta}_{k}(t)\|_{1,k}=0,$$
yielding the uniqueness of the solution to ABH on the interval $[0,T)$ for this peculiar choice of $T=T(\varrho_{0},\alpha).$ However, this procedure may be iterated since $\bm{\beta}_{k+1}(s)=0$ for any $s \in [0,T),$ $k \geq 1$, and it allows to  write Eq. \eqref{eq:betak} as
\begin{equation*}
\langle \bm{\beta}_{k}(t),\Phi_{k}\rangle_{k}=\int_{T}^{t}\langle \bm{\beta}_{k+1}(s)\,,\,\Gk^{\alpha}\Phi_{k}\rangle_{k+1}\d s, \qquad t \geq T, \qquad \forall k \geq 1.\end{equation*}
The same computations that lead to \eqref{Eq:betak-n} now yield
$$\left\|\bm{\beta}_{k}(t)\right\|_{1,k} \leq \frac{t^{n+2}}{(n+2)!}\left((2-\alpha)\|\Sigma_{\B}\|_{\infty}\right)^{n+1}\,\frac{(k+n)!}{(k-1)!}\,\sup_{s\in [0,2T)}\left\|\bm{\beta}_{1,k+n+1}(s)\right\|_{1,k+n+1},$$
for any $t \in [0,2T)$. Then, arguing as before, 
using \eqref{eq:TVnu} which is valid globally, we obtain that
$$\sup_{t \in [0,2T)}\|\bm{\beta}_{k}(t)\|_{1,k} \leq 2^{k-1}T\,\varrho_{0}^{k}\,a^{n+1}, \qquad k \geq 1,\,n\geq 1.$$
Letting $n \to \infty$, we get $\sup_{t \in [0,2T)}\|\bm{\beta}_{k}(t)\|_{1,k}=0$ and, iterating again, we obtain that, for any $k \geq 1$
$$\sup_{t \geq 0}\|\bm{\beta}_{k}(t)\|_{1,k}=0$$
which proves the global uniqueness of solution to ABH.\end{proof}

We now have all we need to prove Theorem \ref{main1}. 
\begin{proof}[Proof of Theorem \ref{main1}] Combining Proposition \ref{theo:existence} and Theorem \ref{unique}, $\bm{\nu}^{\infty}=\{\nu_{k}\}_{k} \in L^{\infty}([0,\infty),\mathcal{X})$ given by 
$$\nu_{k}(t)= f(t)^{\otimes k}, \qquad t \geq 0, \qquad k \geq 1$$ 
is \emph{the unique} solution to the ABH in the sense of Definition \ref{def:weaksol} associated to the initial datum $\nu_{k}(0)=f_{0}^{\otimes k}$, $k \geq 1$. We conclude then with Theorem \ref{main}.
\end{proof}

\section{Some perspectives towards the extension to hard-sphere collision kernel}\label{sec:HS}

As already discussed in the Introduction, the restriction to kernels such that $\Sigma_{\B}$ is bounded -- though not restrictive to the Maxwellian case -- is a severe restriction on our result. The most challenging case would consist in showing the propagation of chaos (Theorem \ref{main1}) for hard-sphere interactions. 
In this section we propose some conjectures in this direction, presenting some intuition for the annihilated model \eqref{BE} and explaining what appears 
to be the major difficulty, as well as some possible paths to overcome it. 

\subsection{Uniqueness for ABH: Failure of De Finetti's approach} \label{ss:DeFinFail}

As we emphasized in the previous sections, the main obstruction in proving the propagation of chaos for more general kernels than the ones satisfying  \eqref{eq:boudedkernel} is in the proof of the uniqueness of solutions to the Annihilated Boltzmann Hierarchy \eqref{sol:BHweak}. Indeed, to recover the propagation of chaos result,  we need to prove the uniqueness of solutions of \eqref{sol:BHweak} with an initial datum of the form $\nu_{k}(0)=f_{0}^{k}$, $k \geq 1$. In other words, the propagation of chaos would follow from the following result which we formulate here as a conjecture:
\begin{conj}\label{conj:unique}
Assume that the collision kernel $\B$ is of the form: 
\begin{equation}\label{eq:Bhardsp}
\B(v-v_{*},\omega)=\frac{1}{2\pi}\left|\left(v-v_{*}\right) \cdot \omega\right| \; \text{for} \;(v,v_{*}) \in \R^{3}\times\R^{3},\omega \in \S
\end{equation}
and denote by $\Q$ the corresponding Boltzmann operator. Let $f_{0}$ be a non-negative probability distribution satisfying \eqref{normalization}. Then, $\bm{\nu}^{\infty}=\{\nu_{k}\}_{k} \in L^{\infty}([0,\infty),\mathcal{X})$ with 
$$\nu_{k}(t)= f(t)^{\otimes k}, \qquad t \geq 0, \qquad k \geq 1$$ 
is \emph{the unique} solution to the ABH in the sense of Definition \ref{def:weaksol} with initial datum $\nu_{k}(0)=f_{0}^{\otimes k}$ for all $k \geq 1$ where $f(t)$ is the unique solution to \eqref{BE} with initial datum $f(0)=f_{0}$.\end{conj}

We notice that we rigorously proved this uniqueness result in the case of $\Sigma_{\B}$ bounded, but we have not been able yet to prove such uniqueness for kernels $\B$ such that $\Sigma_{\B}$ is unbounded as, for instance, in the case of hard-sphere interactions \eqref{eq:Bhardsp}.

We recall that in the classical context, when there is no annihilation in the model (i.e. $\alpha=0$), this problem found a very elegant solution (which can be traced back to \cite{caprino}) using as a powerful tool the De Finetti's Theorem.  
(See also \cite{spohn} for a first use of De Finetti's theorem in a kinetic framework). 
Indeed, when $\alpha=0$, the stochastic particle model (Kac's model), as well as the corresponding limit hierarchy (Boltzmann Hierarchy), is conservative. In particular, any solution $\bm{\nu}^{\infty}$ to the Boltzmann hierarchy satisfies the following 
properties:
\begin{enumerate}
\item Conservation of mass, i.e.
$$\int_{\R^{3\ell}}\nu_{\ell}(t,\d\Vk)=\int_{\R^{3\ell}}\nu_{\ell}(0,\Vk)\d\Vk=\int_{\R^{3\ell}}f_{0}^{\ell}(\Vk)\d\Vk=\varrho_{0}^{\ell}, \qquad \forall t \geq 0,\, \ell \geq 1$$
and there is no loss of generality then to assume that $\varrho_{0}=1.$
\item Compatibility, i.e
$$\Pi_{\ell}\nu_{\ell+1}(t)=\nu_{\ell}(t), \qquad \qquad \forall t \geq 0,\, \ell \geq 1$$
where we recall that $\Pi_{\ell}$ is the marginalisation operator.\end{enumerate}
In the same way, the conservation of mass and the compatibility property are satisfied by the (rescaled) correlation functions, i.e. $\int_{\R^{3\ell}}f_{\ell}^{\varepsilon}(t,\Vk)\d\Vk=\varrho_{0}^{\ell}$ and 
\begin{equation}
\label{eq:compatibile}\Pi_{\ell}f_{\ell+1}^{\varepsilon}(t,\Vk)=f^{\varepsilon}_{\ell}(t,\Vk), \qquad \qquad \forall t \geq 0,\:\:\varepsilon >0, \: \ell \geq 1.\end{equation}

As a consequence, by virtue of De Finetti's Theorem (or Hewitt-Savage's Theorem \cite{hewitt}), one can associate any solution $\bm{\nu}^{\infty}(t)$ of the Boltzmann Hierarchy to a unique probability measure $\pi_{t}$ over the space of probability measures $\mathsf{P}(\R^{3})$ such that
$$\int_{\mathsf{P}(\R^{3})}\mathsf{p}^{\otimes \ell}\,\pi_{t}(\d \mathsf{p})=\nu_{\ell}(t), \qquad \forall \ell \geq 1, \qquad t \geq 0$$
where, for any $\mathsf{p} \in \mathsf{P}(\R^{3})$, $\mathsf{p}^{\otimes \ell}$ is the probability measure over $\R^{3\ell}$ defined by 
$$\int_{\R^{3\ell}}\Phi_{\ell}(\Vk)\mathsf{p}^{\otimes \ell}(\d\Vk)=\prod_{j=1}^{\ell}\int_{\R^{3}}\phi_{j}(v)\mathsf{p}(\d v), \qquad \forall \Phi_{\ell}=\overset{\ell}{\underset{i=1}{\otimes}}\phi_{i}  \in \mathscr{C}_{b}(\R^{3\ell})$$
i.e. $\Phi_{\ell}(\Vk)=\phi_{1}(v_{1})\phi_{2}(v_{2})\ldots\phi_{\ell}(v_{\ell})$, $\phi_{j} \in \mathcal{C}_{b}(\R^{3}).$ Notice that the class of $\Phi_{\ell}$ of this form is dense in $\mathscr{C}_{b}(\R^{3\ell})$ thanks to Stone-Weierstrass theorem. 

In particular, the mapping $t \geq 0 \longmapsto \pi_{t} \in \mathsf{P}(\mathsf{P}(\R^{3}))$ is a solution to the following hierarchy of equations (see Eq. \eqref{sol:BHweak2}):
\begin{multline}\label{eq:Hierpi}
\int_{\mathsf{P}(\R^{3})}\langle \mathsf{p}^{\otimes k}, \Phi_{k}\rangle_{k}\pi_{t}(\d \mathsf{p})
=\langle f_{k}(0)^{\otimes k}, \Phi_{k}\rangle_{k}+
  \int_{0}^{t}\d s \int_{\mathsf{P}(\R^{3})} \langle \mathsf{p}^{\otimes (k+1)}, \Gk^{0}\Phi_{k} \rangle_{k+1}\d \pi_{s}(\mathsf{p}) \\
  \qquad \forall t \in [0,T),\quad k \geq 1, \qquad \Phi_{k}=\phi_{1}\otimes \cdots \otimes \phi_{k} \in \C_{0}(\R^{3k})
\end{multline}
where now $\Gk^{0}$ is the operator corresponding to $\alpha=0.$ Then, using again De Finetti's Theorem, it is possible to show that the above  hierarchy of equations \eqref{eq:Hierpi} has a \emph{unique} solution $t \geq 0\longmapsto \pi_{t} \in \mathsf{P}(\R^{3})$  and this yields the desired uniqueness for the Boltzmann Hierarchy in the case of no annihilation, i.e. $\alpha=0$. (See \cite{caprino,MM} for details).

As soon as we introduce the kinetic annihilation, i.e.~we consider $\alpha \in (0,1]$, it appears very difficult to adapt in a direct way the approach proposed for the Boltzmann Hierarchy. We first observe that solutions to ABH are not probability measures because of the loss of mass. We could expect that a suitable adaptation of De Finetti's Theorem may hold for compatible sequences in $\mathcal{X}$ rather than in $\prod_{k=1}^{\infty}\mathsf{P}(\R^{3k})$ provided we have some control of the dissipation of total mass. Nonetheless, 
even in this case and \emph{because of the mass dissipation}, the (rescaled) correlation functions $f_{\ell}^{\varepsilon}(t,\Vk)$ are  \emph{not even compatible}, i.e. \eqref{eq:compatibile} is not valid for $\alpha \neq 0$. Consequently,  we can not assume \emph{a priori} that solutions to the ABH are compatible \footnote{Notice that the special solution $\nu_{k}(t)=f(t)^{\otimes k}$ is compatible and, since we expect it to be the unique solution, we believe that solutions to \eqref{sol:BHweak} are compatible.}. Therefore, it appears not clear at all how to adapt De Finetti's argument to the case of \emph{non compatible} sequences. 

\subsection{Perturbative argument: the case $\alpha \simeq 0$} 
Due to the difficulty raised from the previous discussion, it could be feasible that Conjecture \ref{conj:unique} is too strong. 
It would be satisfactory to prove propagation of chaos not globally in time but on a suitable time interval (as it appears for instance in Lanford's result for the non homogeneous Boltzmann equation, see \cite{simonella} and the references therein) and for a \emph{moderate} annihilation. This leads us to reformulate Conjecture \ref{conj:unique} in a weaker form as follows.

\begin{conj}\label{conj:unique1}
Assume that the collision kernel $\B$ is of the form: 
\begin{equation*}
\B(v-v_{*},\omega)=\frac{1}{2\pi}|\left(v-v_{*}\right) \cdot \omega \; \text{for} \;(v,v_{*}) \in \R^{3}\times\R^{3},\omega \in \S
\end{equation*}
and denote as $\Q$ the corresponding Boltzmann operator. Let $f_{0}$ be a non-negative probability distribution satisfying \eqref{normalization}. Then, there exists $\alpha_{0} \in (0,1]$ such that, for any $\alpha \in (0,\alpha_{0}),$ there is $T_{\alpha} >0$ such that the Annihilated Boltzmann Hierarchy \eqref{sol:BHweak} admits a \emph{unique} weak solution 
$$\bm{\nu}^{\infty}=\{\nu_{k}\}_{k} \in L^{\infty}([0,T_{\alpha}),\mathcal{X})$$
with initial datum $\nu_{k}(0)=f_{0}^{\otimes k}$ for all $k \geq 1$. In particular, $\nu_{k}(t)= f(t)^{\otimes k},$ for any $t \in [0,T_{\alpha}), k \geq 1$  where $f(t)$ is the unique solution to \eqref{BE} with initial datum $f(0)=f_{0}$.
\end{conj}
 
It is interesting to observe that it seems likely that Conjecture \ref{conj:unique1} could be proved using a perturbative argument, looking at the case $\alpha \simeq 0$ as perturbation of the more handable case $\alpha=0$. 
More precisely, it could be possible to exploit the fact that for $\alpha=0$ there is a unique solution to the Boltzmann Hierarchy, as explained in Section \ref{ss:DeFinFail}, and obtain   
the uniqueness of a solution to \eqref{sol:BHweak} for $\alpha$ small enough and for a 
possibly small time interval $[0,T_{\alpha})$. 
For instance, we could think to apply the perturbative argument not necessarily at the level of the Annihilated Boltzmann Hierarchy but at the level of the hierarchy of equations \eqref{eq:Hierpi} and try to prove that, for $\alpha$ positive but small enough, the perturbed hierarchy 
\begin{multline}\label{eq:Hierpialpha}
\int_{\mathsf{P}(\R^{3})}\langle \mathsf{p}^{\otimes k}, \Phi_{k}\rangle_{k}\pi_{t}^{\alpha}(\d \mathsf{p})
=\langle f_{k}(0)^{\otimes k}, \Phi_{k}\rangle_{k}+
  \int_{0}^{t}\d s \int_{\mathsf{P}(\R^{3})} \langle \mathsf{p}^{\otimes (k+1)}, \Gk^{\alpha}\Phi_{k} \rangle_{k+1} \pi_{s}^{\alpha}(\d \mathsf{p}) \\
  \qquad \forall t \in [0,T),\quad k \geq 1, \qquad \Phi_{k}=\phi_{1}\otimes \cdots \otimes \phi_{k} \in \C_{0}(\R^{3k})
\end{multline}
still admits a unique solution  $\pi_{t}^{\alpha} \simeq \pi_{t}$ for $\alpha \simeq 0$ (in a suitable sense to be clarified). Such a solution $\pi_{t}^{\alpha}$ has then to be interpreted in terms of (statistical) solutions to \eqref{sol:BHweak}. To prove this, we should exploit the fact that
$$\Gk^{\alpha}\Phi_{k} \simeq \Gk^{0}\Phi_{k}, \qquad \text{ for } \alpha \simeq 0, \quad k \geq 1.$$
Notice however that the implementation of this approach 
 is again non trivial -- essentially because \eqref{sol:BHweak} does not admit a unique solution in the full space $\mathcal{X}$ but only in the subclass of \emph{compatible} sequences. As already observed, this subclass does not seem to be stable in the limit $\alpha \simeq 0.$

\subsection{Self-similar variables}

Let us consider a solution $f=f(t,v)$ to \eqref{BE} for some nonnegative initial datum $f_{0} \in L^{1}_{3}(\R^{d}).$ It has been shown in \cite[Proposition 1.2]{ABL} that, introducing $\psi(\tau,\xi)$ through
\begin{equation}\label{scalingPsi}
f(t,v)=n_{f}(t){(2T_{f}(t))^{-3/2}}\psi\left(\tau(t),\frac{v-\bm{u}_{f}(t)}{\sqrt{2T_{f}(t)}}\right)
\end{equation} 
with $n_{f}(t)=\displaystyle\int_{\R^{3}}f(t,v)\d v$,   $n_{f}(t)\bm{u}_{f}(t)=\displaystyle\int_{\R^{3}}f(t,v)v\d v,$ and
$$3\,n_{f}(t)T_{f}(t)=\int_{\R^{3}}f(t,v)|v-\bm{u}_{f}(t)|^{2}\d v, \qquad \tau(t)=\sqrt{2}\int_{0}^{t} n_{f}(s)\sqrt{T_{f}(s)}\d s, \qquad t \geq0.$$then it holds that $\psi(\tau,\xi)$ 
is the unique solution to 
\begin{align}\label{rescaBE}
\begin{split}
\partial_{\tau}\psi(\tau,\xi) + \big(\mathbf{A}_{\psi}(\tau) - &d \mathbf{B}_{\psi}(\tau)\big)\,\psi(\tau,\xi) + \mathbf{B}_{\psi}(\tau)\mathrm{div}_{\xi}\big(\left({\xi}-\bm{v}_{\psi}(\tau)\right)\psi(\tau,\xi))\\
&= (1-\alpha)\Q(\psi,\psi)(\tau,\xi)-\alpha\Q_{-}(\psi,\psi)(\tau,\xi)\end{split}\end{align}
with initial datum $\psi(0,\xi)=(2T_{f_{0}})^{3/2}n_{f_{0}}^{-1}\;f_{0}\left(\sqrt{2T_{f_{0}}}\,\xi+\bm{u}_{f_{0}}\right)$ and 
where $\mathbf{A}_{\psi}(\cdot),\mathbf{B}_{\psi}(\cdot)$ and $\bm{v}_{\psi}(\cdot)$ are defined by 
\begin{equation*}
\left(\begin{array}{c}\mathbf{A}_\psi(\tau)\\  \mathbf{B}_{\psi}(\tau) \\  \mathbf{B}_{\psi}(\tau)\,\bm{v}_{\psi}(\tau)\end{array}\right)=-\frac{\alpha}{2}\displaystyle\int_{\R^{3}} \Q_-(\psi ,\psi )(t,\xi)\left(\begin{array}{c} 5 
-2|\xi|^{2}\\  1-\frac{2}{3}|\xi|^2 \\  2\xi\end{array}\right)\d\xi, \qquad \forall \tau \geq 0.
\end{equation*}
The interesting feature of the above equivalent kinetic annihilation model is that, in contrast to \eqref{BE}, it is \emph{conservative}, i.e.
\begin{equation}\label{eq:conserve}
\int_{\R^{d}}\psi(\tau,\xi)\left(\begin{array}{c}1\\\xi \\ {|\xi|^{2}}\end{array}\right)\d\xi=\int_{\R^{d}}\psi(0,\xi)\left(\begin{array}{c}1\\\xi \\ {|\xi|^{2}}\end{array}\right)\d\xi=\left(\begin{array}{c}1 \\0 \\\frac{d}{2}\end{array}\right) \qquad \forall \tau \geq0.\end{equation} 
As observed in Section \ref{ss:DeFinFail}, the dissipation of mass for \eqref{BE} is the main obstacle in obtaining the uniqueness for ABH for more general kernel, using  De Finetti's approach. Due to the equivalence between \eqref{rescaBE} and \eqref{BE}, a possible way to overcome this obstacle is to derive \eqref{rescaBE} from a particle system. In this case, it would be possible to exploit the canonical ensemble formalism (instead of the grand canonical approach followed in this paper) but the prize to pay would be to derive a \emph{non-autonomous} Kac's model. Another technical difficulty seems to be that, as it is apparent from the form of $\mathbf{A}_{\psi},\mathbf{B}_{\psi}$, the above \eqref{rescaBE} is actually a \emph{trilinear equation} for the unknown $\psi(\tau,\xi).$ Attempts to derive trilinear kinetic equations from a particle system seem highly technical.

\appendix
\section{Properties of the semigroup $\left(\mathcal{S}_N(t)\right)_{t \geq 0}$} \label{appendix1}
 
We establish here, for a given $N \geq 1$, several properties of the $C_{0}$-semigroup $\left(\mathcal{S}_{N}(t)\right)_{t \geq 0}$ generated by $\mathcal{L}_{N}.$ Since, in all this section, $N \geq 1$ is fixed, we will simply write $\Sigma=\sigma_{N}$, $\mathcal{L}=\mathcal{L}_{N}$, $\mathcal{S}(t)=\mathcal{S}_{N}(t)$ and set $d=3N$. 
Therefore, one has
$$\Sigma=\Sigma(\V)=\sum_{1\leq i < j \leq N}\int_{\S}\mathcal{B}(v_{i}-v_{j},\omega)\d\omega=\sum_{1\leq i < j \leq N}|v_{i}-v_{j}| \qquad \forall \V=(v_{1},\ldots,v_{N})\in \R^{d}.$$
We define for simplicity
$$L^{1}=L^{1}(\R^{d}) \qquad \text{ and } \qquad L^{1}_{\Sigma}=L^{1}(\R^{d},\Sigma(V)\d V).$$
The multiplication operator
$$T\::\D(T) \subset L^{1}(\R^{d}) \to L^{1}(\R^{d})$$
defined by
$$T\Phi (\V)=-\Sigma(\V)\Phi(\V)\,\qquad \forall \Phi \in \D(T)=L^{1} \cap L^{1}_{\Sigma}$$
is the generator of a positive $C_{0}$-semigroup $\left({U}_{0}(t)\right)_{t \geq 0}$ in $L^{1}$ given by
\begin{equation}\label{U0}
U_{0}(t)\Phi (\V)=\exp\left(-\Sigma(\V)t\right)\Phi(\V) \qquad \V \in \R^{d}.\end{equation}
Since $x\exp(-x) \leq 1/e$ for any $x \geq 0$, one easily gets that
$$\|U_{0}(t)\Phi\|_{L^{1}_{\Sigma}} \leq \dfrac{1}{et}\,\|\Phi\|_{L^{1}} \qquad \forall t > 0.$$
Moreover, for any $\Phi \in L^{1}$ and $T > 0$, one also has
\begin{multline*}
\int_{0}^{T}\|U_{0}(t)\Phi\|_{L^{1}_{\Sigma}}\d t=\int_{0}^{T}\d t\int_{\R^{d}}\left|\Phi(\V)\,\Sigma(\V)\exp\left(-t\Sigma(\V)\right)\right|\d \V\\
=\int_{\R^{d}}\Sigma(V)\,|\Phi(\V)|\,\d \V\int_{0}^{T}\exp\left(-t\Sigma(\V)\right)\d t=\int_{\R^{d}}|\Phi(\V)|\left(1-\exp(-T\Sigma(\V))\right)\d \V
\end{multline*}
so that
\begin{equation}\label{smo}
\int_{0}^{T}\|U_{0}(t)\Phi\|_{L^{1}_{\Sigma}}\d t\leq \|\Phi\|_{L^{1}} \qquad \forall \Phi \in L^{1},\;\;T > 0.
\end{equation}
Moreover, for any $\lambda >0$, the resolvent $(\lambda-T)^{-1}$ of $T$ satisfies
$$\|(\lambda-T)^{-1}\Phi\|_{L^{1}_{\Sigma}} \leq \sup_{\V \in \R^{d}}\dfrac{\Sigma(\V)}{\lambda+\Sigma(\V)}\|\Phi\|_{L^{1}} =\|\Phi\|_{L^{1}}.$$
Let us introduce the linear operator $K=(1-\alpha)\mathbf{G}_{N}$. Then, since 
$$\|\mathbf{G}_{N}\Phi\|_{L^{1}}=\|\Phi\|_{L^{1}_{\Sigma}} \qquad \forall \Phi \in L^{1}_{\Sigma}, \:\Phi \geq 0$$
it follows that 
$$\|K(\lambda-T)^{-1}\Phi\|_{L^{1}} \leq (1-\alpha)\|\Phi\|_{L^{1}} \qquad \forall \lambda >0\,\forall \Phi \in L^{1}.$$
Therefore, as soon as $\alpha \in (0,1)$, one deduces from Desch's Theorem (see for instance \cite{m2As, MMK}) that 
$$\mathcal{L}=\mathcal{L}_{N}=K+T, \qquad \D(\mathcal{L})=\D(T)$$
is the generator of a positive $C_{0}$-semigroup $\left(\mathcal{S}(t)\right)_{t \geq 0}$ of contractions given by the Dyson-Phillips expansion series:
$$\mathcal{S}(t)\Phi=\sum_{{j=0}}^{\infty}U_{j}(t)\Phi$$
with $U_{0}(t)$ defined by \eqref{U0} and
$$U_{j+1}(t)\Phi=\int_0^{t}U_{j}(t-s)KU_{0}(s)\Phi\d s; \qquad \Phi \in L^{1}\;;\,j \geq 0.$$
It is easy to check that $U_{j}(t)$ inherits the smoothing property \eqref{smo}, namely
\begin{lem}
For any $T > 0$ and any $j \geq 0$, one has
\begin{equation}\label{Uj}\int_{0}^{T}\|U_{j}(t)\Phi\|_{X_{\Sigma}}\d t \leq (1-\alpha)^{j}\,\|\Phi\|_{X} \qquad \forall \Phi \in X.\end{equation}
\end{lem}
\begin{proof} We argue by induction. We observe that \eqref{smo} establishes the wished property for $j=0$. For a given $j \geq 0$, let us assume \eqref{Uj} holds true. Then, from the definition of $U_{j+1}(t)$ one has
\begin{multline*}
\int_{0}^{T}\|U_{j+1}(t)\Phi\|_{L^{1}_{\Sigma}}\d t\leq \int_{0}^{T}\d t\int_{0}^{t}\|U_{j}(t-s)KU_{0}(s)\Phi\|_{L^{1}_{\Sigma}}\d s\\
=\int_{0}^{T}\d s \int_{s}^{T}\|U_{j}(t-s)KU_{0}(s)\Phi\|_{L^{1}_{\Sigma}}\d t\\
\leq \int_{0}^{T}\d s \int_{0}^{T}\|U_{j}(\tau)KU_{0}(s)\Phi\|_{L^{1}_{\Sigma}}\d \tau.
\end{multline*}
For any $s \geq 0$, applying the induction hypothesis \eqref{Uj} to $KU_{0}(s)\Phi$ one gets
$$\int_0^{T}\|U_{j+1}(t)\Phi\|_{L^{1}_{\Sigma}}\d t \leq (1-\alpha)^{j}\int_{0}^{T}\|KU_{0}(s)\Phi\|_{L^{1}}\d s$$
and, since $\|K\psi\|_{L^{1}} \leq (1-\alpha)\|\psi\|_{L^{1}_{\Sigma}}$ for any $\psi \in L^{1}_{\Sigma}$, one gets
$$\int_{0}^{T}\|U_{j+1}(t)\Phi\|_{L^{1}_{\Sigma}}\d t \leq (1-\alpha)^{j+1}\,\int_{0}^{T}\|U_{0}(s)\Phi\|_{L^{1}_{\Sigma}}\d s.$$
Using again \eqref{smo}, one sees that \eqref{Uj} is satisfied by $U_{j+1}(t)$ which achieves the proof.
\end{proof}

\begin{proof}[Proof of Proposition \ref{propLN}]
The proof of \eqref{smoothing} is now a direct consequence of \eqref{Uj} since, for any $T >0$ and any $\Phi \in L^{1}$, the following inequality holds:
\begin{equation*}
\int_{0}^{T}\|\mathcal{S}(t)\Phi\|_{L^{1}_{\Sigma}}\d t \leq \sum_{j=0}^{\infty}\int_{0}^{T}\|U_{j}(t)\Phi\|_{L^{1}_{\Sigma}}\d t
\leq \sum_{j=0}^{\infty}(1-\alpha)^{j}\|\Phi\|_{L^{1}}={\alpha}^{-1}\|\Phi\|_{L^{1}}.
\end{equation*}
This gives the result.
\end{proof}

\section*{Acknowledgments} B. L and A. N. acknowledge support from the \textit{de Castro Statistics Initiative}, Collegio
C. Alberto, Torino, Italy. A.N. acknowledges support through the CRC 1060 \emph{The mathematics of emergent effects} of the University of Bonn that is funded through the German
Science Foundation (DFG).

\end{document}